\documentclass{article}

\usepackage{float}   
\usepackage{graphicx}  
\usepackage{hyperref}  
\usepackage{latexsym}  
\usepackage{xcolor}   
\usepackage{lipsum}   
\usepackage{subcaption}

\usepackage{amssymb}  
\usepackage{amsmath}  
\usepackage{amsbsy}
\usepackage{amsthm}
\usepackage{bbm}    
\usepackage{eucal}   
\usepackage{mathabx}  
\usepackage{mathrsfs}  

\usepackage{multirow}  
\usepackage{ulem}    

\usepackage[numbers,sort&compress]{natbib} 

\usepackage{enumerate} 
\usepackage{enumitem}  

\newtheorem{theorem}{Theorem}[section]

\newtheorem{definition}[theorem]{Definition}

\newtheorem{proposition}[theorem]{Proposition}

\newtheorem{lemma}[theorem]{Lemma}

\newtheorem{remark}{Remark}
\newtheorem{assumption}{assumption}

\newtheorem{corollary}{Corollary}[theorem]

\allowdisplaybreaks

\def\wt{\widetilde}
\def\wh{\widehat}
\def\wb{\overline}
\def\ul{\underline}
\newcommand{\1}{\mathbbm{1}} 
\def\E{\mathbb{E}} 

\newcommand{\norm}[1]{{\vert\kern-0.25ex\vert #1 \vert\kern-0.25ex\vert}}
\newcommand{\bignorm}[1]{{\big\vert\kern-0.25ex\big\vert #1 \big\vert\kern-0.25ex\big\vert}}
\newcommand{\opnorm}[1]{{\vert\kern-0.25ex\vert\kern-0.25ex\vert #1 \vert\kern-0.25ex\vert\kern-0.25ex\vert}}

\def\bbZ{\mathbb{Z}}

\newcommand{\calC}{\mathcal{C}} 
 \newcommand{\calF}{\mathcal{F}}

\usepackage{fancyhdr}

\title{Passive Market Impact: A Point Process Approach
}

\author
{Youssef Ouazzani Chahdi\footnote{MICS, CentraleSupélec, \texttt{youssef.ouazzani-chahdi@centralesupelec.fr}} \and Mathieu Rosenbaum\footnote{CMAP, Ecole Polytechnique, \texttt{mathieu.rosenbaum@polytechnique.edu}} \and Gr\'egoire Szymanski\footnote{DMATH, Université du Luxembourg, \texttt{gregoire.szymanski@uni.lu}}}

\begin{document}

\date{\today}

\maketitle

%
%
%

\begin{abstract}
While the market impact of aggressive orders has been extensively studied, the impact of passive orders—those executed through limit orders—remains less understood. The goal of this paper is to investigate passive market impact by developing a microstructure model connecting liquidity dynamics and price moves. A key innovation of our approach is to replace the traditional assumption of constant information content for each trade by a function that depends on the available volume in the limit order book. Within this framework, we explore scaling limits and analyze the market impact of passive metaorders. Additionally, we derive useful approximations for the shape of market impact curves, leading to closed-form formulas that can be easily applied in practice.
\end{abstract}



\section{Introduction}

In financial markets, traders generally place their orders using market orders, also referred to as liquidity-taking trades, and limit orders, also known as liquidity-providing trades. As explained in Bouchaud et al. \cite{bouchaud2003fluctuations},  market orders are placed by liquidity takers who seek immediate execution, which comes at the expense of paying the bid-ask spread. These orders often exhibit long-range persistence due to strategies like order fragmentation, where a single large order, termed a metaorder, is broken into smaller, sequentially executed trades. Conversely, limit orders are used by liquidity providers, including market makers. Unlike market orders, limit orders do not guarantee immediate execution but instead rest in the order book until matched by a corresponding market order. This allows liquidity providers to avoid taking outright directional positions in the market. However, the traditional distinction between liquidity takers and liquidity providers is diminishing. Nowadays, market participants often engage in both roles, and algorithmic trading systems optimize in real time the choice between submitting limit orders and market orders.

Market impact refers to the fact that buy orders push on average the price up and sell orders push it down. Market impact stands out as a prominent transaction cost associated with the execution of metaorders, see Freyre-Sanders et al. \cite{freyre2004review}, Almgren et al. \cite{almgren2005direct}, Engle et al. \cite{robert2012measuring}, and Hey et al. \cite{hey2023cost}. The measurement and understanding of market impact have thus emerged as a central theme in quantitative finance, see Webster \cite{webster2023handbook} for a review. Qualitatively, market impact can be understood as a way to pass on information to the price: investors decide their strategy using a long term return on their investment anticipation and decide to rebalance their position accordingly using metaorders. In Gomes and Waelbroeck \cite{gomes2015market}, it is shown that the permanent impact of cash flow trades (that is trades that do not originate a particular directional view of the trader on the market, but are rather done for mechanical reasons such as hedging, risk management...) is negligible while the impact of informed trades remains large.

Market impact is inherently difficult to measure because of its noisy nature. Statistical studies usually focus more on the execution of metaorders, which induces a notable liquidity imbalance that results in price moves which can be statistically identified. However, during the execution of a given metaorder, many other orders are likely being traded simultaneously which creates a significant noise. Using a very careful statistical treatment, averaging over many metaorders eliminates part of this noise so that we can better identify universal properties of market impact, see for instance Almgren et al. \cite{almgren2005direct}, Bershova and Rakhlin \cite{bershova2013non}, Bacry et al. \cite{bacry2015market}, and Bucci et al. \cite{bucci2019crossover}. When trading a metaorder, these empirical studies show that the price mechanically follows the metaorder, exhibiting a concave shape peaking at the end of the metaorder, followed by a convex relaxation, see Bershova and Rakhlin \cite{bershova2013non}, Gatheral \cite{gatheral2010no}, and Moro et al. \cite{moro2009market}. During the execution of the metaorder, Almgren et al. \cite{almgren2005direct},  Hopman \cite{hopman2003essays}, Kyle and Obizhaeva \cite{kyle2023large}, Lillo et al. \cite{lillo2003master}, and Moro et al. \cite{moro2009market} identified a square root dependence in the volume. However, this square root dependence only holds for large volumes while the impact of small orders is proportional to the volume, see Benzaquen and Bouchaud \cite{benzaquen2018market} for instance where the authors proposed a more accurate approximation of the market impact of the form
\begin{equation}
\label{eq:market_impact_volume}
MI(Q_t) \approx c \sigma \bigg( \frac{Q_t}{V} \bigg)^{1/2} \mathcal{F}\bigg(\frac{Q_t}{V} \bigg),
\end{equation}
where $c$ is a constant of order 1, $\sigma$ is the daily volatility, $V$ the usual daily volume traded on the asset in normal conditions, $Q_t$ is the volume executed by a given metaorder at time $t$ and where $\mathcal{F}$ is monotonic and satisfies $\mathcal{F}(x) \approx \sqrt{x}$ when $x\to0$ and $\mathcal{F}(x) \to a$ when $x\to\infty$ for some $a > 0$.

Traditionally, most metaorders are executed using market orders, which has led to the development of quantitative models for market impact based on the market order flow, see for instance Jaisson \cite{jaisson2015market}. These models leverage the idea that the price response to a market order is mechanical: the execution of a market order creates an imbalance between the bid and the ask, ultimately resulting in a price change. Within this framework, it is natural to define the price as the expectation of the future order flow. Specifically, let $N^a_t$ (resp. $N^b_t$) denote the number of buy (resp. sell) market orders before time $t$. Then, the price is given by
\begin{equation}
\label{eq:price}
P_t 
= 
P_0 +
\lim\limits_{s\to\infty} \kappa  \mathbb{E}[N^{a}_s - N^{b}_s  |  \mathcal{F}_t ],
\end{equation}
where $ \mathcal{F}_t$ denotes the information available at time $t$. This expectation requires a suitable probabilistic model on $N^a$ and $N^b$, though it can still be understood as an empirical expectation from the perspective of the market participants. Here, $\kappa > 0$ is a constant that represents the price move created by a single order, effectively quantifying the amount of information conveyed in one given trade. In this simplified model, with all trades being implicitly considered equal, it is reasonable to assume that $\kappa$ is constant.

In recent years, Hawkes processes have proven to be a natural framework for modeling transaction arrival times, see Filimonov and Sornette \cite{filimonov2012quantifying,filimonov2015apparent}, Hardiman et al. \cite{hardiman2013critical}, and Jaisson and Rosenbaum \cite{jaisson2015limit,jaisson2016rough}. These processes were introduced in the context of price impact modeling by Jaisson \cite{jaisson2015market}, where $N^a$ and $N^b$ are assumed to be two independent Hawkes processes with self-exciting kernel $\varphi$. It is shown in the same work that \eqref{eq:price} can be computed explicitly, yielding
\begin{equation}
\label{eq:price:hawkes}
P_t 
= 
P_0 +
\lim\limits_{s\to\infty} \kappa  \mathbb{E}[N^{a}_s - N^{b}_s  |  \mathcal{F}_t ]
=
P_0 + \kappa 
\int_0^t 
\xi(t-s)  d (N^{a}_s - N^{b}_s)
\end{equation}
where $\xi(t) = 1 + ( 1 + \int_0^\infty \psi(s)  ds ) \int_t^\infty \varphi(s)  ds$ and $\psi = \sum_{k\geq 1} \varphi^{*k}$ where $\varphi^{*k}$ stands for the $k$ fold convolution of $\varphi$. In particular, note that $\xi(t) \to 1$ as $t \to \infty$ and therefore the permanent impact of a single order is given by $\kappa$. Note also that Equation \eqref{eq:price:hawkes} is a special case of propagator models, generalising the discrete-time propagator proposed in Bouchaud et al. \cite{bouchaud2003fluctuations}. This model was extensively studied in Jusselin and Rosenbaum \cite{jusselin2020noarbitrage}, and Durin et al. \cite{durin2023two} where it is shown that it is consistent with the characteristic shape of market impact given in \eqref{eq:market_impact_volume}. 

This yields the following question: What happens if an investor executes a given metaorder through limit orders instead of market orders? Using limit orders yields better execution prices, at the cost of a non-execution risk. Various contributions quantify the trade-off between price improvement and non-execution risk when replacing market orders with limit orders, including Gu\'eant et al. \cite{gueant2012optimal}, Guilbaud and Pham \cite{guilbaud2013optimal}, and Bayraktar and Ludkovski \cite{bayraktar2014liquidation}, while Cartea and Jaimungal \cite{cartea2013modelling}, Cont and Kukanov \cite{cont2017optimal}, and Guo et al. \cite{guo2017optimal} develop detailed models for the optimal placement of limit orders within the order book.

Several empirical studies have confirmed that the market impact of limit orders is comparable to that of market orders , see Eisler et al. \cite{eisler2012price}, and Said et al. \cite{said2017market}. However, intriguingly, propagator models such as \eqref{eq:price:hawkes} seem to create a price impact of the opposite sign for limit orders. Take for instance a buy metaorder that is passed through limit orders. Then, it would be available on the bid side, and therefore, when it is executed, the price would go down on average. This is a misconception because in our Hawkes model, market orders arrival times are independent of the current limit order book state, and therefore market orders would be unchanged by the presence of these limit orders. This is of course an unrealistic assumption and many works confirmed that the limit order book state influenced future price changes. Concretely, several effects interplaying in the price formation process create a price impact on the correct side for limit orders. First, the presence of an additional order on the limit order book creates an imbalance that means that it is harder for the price to move in that direction. This fact is confirmed for instance by empirical studies that show that the sign of the next price movement can be predicted by the current imbalance of the order book, see Burghardt et al. \cite{burghardt2006measuring}, Gould and Bonart \cite{gould2016queue}, Lehalle and Mounjid \cite{lehalle2017limit}, and Stoikov \cite{stoikov2018micro}. Moreover, the intensity at which market orders arrive depends on the volume available at the best bid and best ask prices, see Huang et al. \cite{huang2015simulating}. This effect is particularly true when a pile is almost empty. In that case, it attracts market orders with a very high probability, provoking a price change. 

In this paper, we propose a novel approach to model the market impact of limit orders. The idea consists in weighting the market order flow to take into account the current limit order book state. Specifically, we denote by $q^a_t$ and $q^b_t$ the ask and bid limit order book states at time $t$. In \eqref{eq:price}, the price jumps of a constant $\kappa$ if a market order is placed at time $t$. Here, we replace this constant $\kappa$ by $\kappa(q^a_t)$ and $\kappa(q^b_t)$ depending on the side of the market order. In Section \ref{sec:model}, we introduce a dynamic for $q^a$ and $q^b$ inspired by the Queue Reactive model in Huang et al. \cite{huang2015simulating} and we set
\begin{equation*}
P_t = P_0 + \lim_{T\to\infty}
\E\bigg[\int_0^T \kappa(q^a_s)  dN^a_s - \int_0^T \kappa(q^b_s)  dN^b_s \bigg| \calF_t \bigg].
\end{equation*}
Although this extension is conceptually straightforward, it introduces several technical challenges, as the volumes available in the ask and bid queues depend on the arrival of market orders. Nevertheless, we demonstrate that the price process is well-defined, see Theorem \ref{thm:finiteness:price:lim} for more details. In this model, $\kappa$ can be interpreted as a price resistance coefficient. This suggests that $\kappa$ should be a decreasing function, reflecting that it is harder for the price to move when available volume is higher. We use this approach to study market impact in this setup, as described in Section \ref{sec:mi}. Although a closed-form solution is not obtainable in this model due to its complexity, we derive approximations for the market impact in Section \ref{sec:applications} when the market is near its stationary state. In this scenario, we show that the average instantaneous impact of a unit-volume limit order is given by
\begin{equation}
\label{eq:immediate_impact}
    \pm \frac{c_\kappa}{c_\lambda} \times \frac{\mu}{1 - \norm{\varphi}_{L^1}}
\end{equation}
where $c_\kappa = \kappa'(q)$ represents the sensibility of the price stickiness with respect to volume changes, $c_\lambda$ quantifies the mean-reversion of the order book and $\frac{\mu}{1 - \norm{\varphi}_{L^1}}$ is the average intensity of a stationary version of the Hawkes process $N^a$.

Building on Jaisson and Rosenbaum \cite{jaisson2016rough}, Jusselin and Rosenbaum \cite{jusselin2020noarbitrage}, and Durin et al. \cite{durin2023two}, we derive in Section \ref{sec:scaling} the scaling limits of this model. This approach enables us to examine the long-term effects of executing a large metaorder in this market. Under additional assumptions, we show that the market impact of the execution of a large metaorder executed through limit orders passed with intensity $f$  is given by
\begin{equation}
\label{eq:scaling_impact}
    \mathrm{MI}_t = c_\kappa \int_0^t \int_0^{s} e^{c_\lambda (s-u)} f(u)  du  Y_s  ds
+ 
c_\kappa \int_t^\infty \int_0^{t} e^{c_\lambda (s-u)} f(u)  du  \E[Y_s|\calF_t]  ds,
\end{equation}
where $c_\kappa$ and $c_\lambda$ are the same as in \eqref{eq:immediate_impact}, $Y_s$ is the market instantaneous variance, and $\E[Y_s|\calF_t]$ denotes the forward variance curve. This formulation, where all variables can be estimated, allows brokers to assess a trading strategy \textit{a posteriori}.

The rest of this paper is structured as follows. In Section \ref{sec:model}, we introduce a new price model incorporating information from the limit order book, with an analysis of market impact in this context presented in Section \ref{sec:mi}. Section \ref{sec:scaling} then explores the scaling limits of the price and market impact models. Section \ref{sec:shape} presents useful approximations, including \eqref{eq:immediate_impact} and \eqref{eq:scaling_impact}, as well as an application for brokers. Finally, Section \ref{sec:estimation} proposes a statistical analysis of $\kappa$.

\section{A new micro-structure price model}
\label{sec:all_models}

\subsection{Price model}
\label{sec:model}

In this section, we define a novel price model based on the Hawkes propagator model introduced in \cite{jaisson2015market} taking into account the current limit order book state. For simplicity sake, we will assume here that all orders have the same volume, and we normalize this volume to one.

We first assume that the buy and sell market order flows follow Hawkes processes as in \cite{jaisson2015market,jusselin2020noarbitrage,durin2023two}. The arrival times of buy and sell market orders are the jump times of two independent Hawkes processes, denoted by $N^a$ (ask, i.e. buy market orders) and $N^b$ (bid, i.e. sell market orders). Both processes have the same baseline intensity $\mu \geq 0$ and self-exciting (non-negative) kernel $\varphi$ so that their intensities are given by
\begin{equation*}
\lambda^{a}_t = \mu + \int_0^{t-} \varphi(t-s)  dN^{a}_s
\;\;\text{ and }\;\;
\lambda^{b}_t = \mu + \int_0^{t-} \varphi(t-s)  dN^{b}_s.
\end{equation*}

We then consider a limit order book. We model the aggregated volumes available on the ask and bid side by $q^a$ and $q^b$. We model their dynamic in a similar fashion as in the Queue Reactive model \cite{huang2015simulating}. More precisely, we write
\begin{equation*}
q^a_t = L^{a}_t - C^{a}_t - N^{a}_t \quad \text{ and } \quad q^b_t = L^{b}_t - C^{b}_t - N^{b}_t
\end{equation*}
where $L^a_t$ and $L^b_t$ are the number of limit orders on the ask and the bid side before time $t$ and $C^a_t$ and $C^b_t$ are the number of cancel orders on the ask and the bid side before time $t$. Furthermore, we assume that $L^{a}$ and $L^{b}$ are two time-inhomogeneous Poisson processes whose intensities are given by 
\begin{equation*}
\lambda_t^{L,a} = \lambda^L(q_{t-}^{a})
\quad \text{ and } \quad 
\lambda_t^{L,b} = \lambda^L(q_{t-}^{b})
\end{equation*}
for some given function $\lambda^L$; and that $C^{a}$ and $C^{b}$ are two time-inhomogeneous Poisson processes whose intensities are given by 
\begin{equation*}
\lambda_t^{C,a} = \lambda^C(q_{t-}^{a})
\quad \text{ and } \quad 
\lambda_t^{C,b} = \lambda^C(q_{t-}^{b})
\end{equation*}
for some given function $\lambda^C$.

In what precedes, we implicitly assume that the Poisson point measures driving $N^a$, $N^b$, $L^a$, $L^b$, $C^a$ and $C^b$ are independent. In other words, we assume that there exist some independent Poisson point measures $\pi^{N,a}$, $\pi^{N,b}$, $\pi^{L,a}$, $\pi^{L,b}$, $\pi^{C,a}$ and $\pi^{C,b}$ such that for $x \in \{a, b\}$, we have
\begin{equation}
\label{eq:def:pointmeasure}
\begin{cases}
N^{x}_t = \int_0^t \int_0^\infty \1_{\{z \leq \lambda^x_s\}}  \pi^{N,x}(ds,dz),
\\
L^{x}_t = \int_0^t \int_0^\infty \1_{\{z \leq \lambda^{L,x}_s\}}  \pi^{L,x}(ds,dz),
\\
C^{x}_t = \int_0^t \int_0^\infty \1_{\{z \leq \lambda^{C,x}_s\}}  \pi^{C,x}(ds,dz).
\end{cases}
\end{equation}
By representing the Poisson process via a Poisson random measure, we may work pathwise. Indeed, for each $\omega$, the realization $\pi(\omega)$ is fixed, and the indicator $1_{\{z \leq \lambda_s\}}$ simply selects the atoms of $\pi(\omega)$ whose mark fall below the intensity $\lambda_s(\omega)$. In this way, once $\pi(\omega)$ is given, the jump times of the Poisson process are identified uniquely and explicitly.

This assumption is equivalent to assuming that jump times of $N^a$, $N^b$, $L^a$, $L^b$, $C^a$ and $C^b$ are almost surely mutually disjoint. We do not prove this statement here, and we do not prove that such a model exists, but we refer to Theorem 2 in \cite{durin2023two} for a proof in a very similar model, and to Jacod \cite{jacod1979calcul} for a proof of the independence between the random measures. Furthermore, we assume that these processes are defined on a rich enough probability space $(\Omega, \calF, (\calF_t)_t, \mathbb{P})$ such that the random measures $(\pi^{N,a}, \pi^{N,b}, \pi^{L,a}, \pi^{L,b}, \pi^{C,a}, \pi^{C,b})$ are adapted to $(\calF_t)_t$, in the sense that their restriction to $[0,t] \times [0,\infty)$ is $\calF_t$-measurable and their restriction to $(t, \infty) \times [0,\infty)$ is independent of $\calF_t$. Therefore, $\calF_t$ can be seen as the information known at time $t$.

Note also that this model allows for negative queue values. For instance, this can happen whenever $q_t^x > 0$, and the Hawkes process $N^x$ jumps $q_t^x + 1$ times before the limit process $L^x$ and the cancel process jumps. Since the intensity of the Hawkes process is always strictly positive, this can happen with strictly positive probability. However, this is not an issue. From a theoretical perspective, the behavior of the market is similar when $q_t^x$ is positive or negative. Moreover, when calibrated to the market $\lambda^L$ and $\lambda^C$ are such that the probability of a negative queue remains small. 

Inspired by previous papers \cite{jaisson2015market,jusselin2020noarbitrage,durin2023two}, we then define the price as the limit of the anticipation of the upcoming market orders with a weight that depends on the limit order book state at the time of execution of the market order. This weight encodes the fact that it is harder for the price to move in one direction when the available volume in this side of the limit order book is higher. Specifically, we introduce for all $t \geq 0$
\begin{equation}
\label{eq:price:lim}
P_t = P_0 + \lim_{T\to\infty}
\E\bigg[\int_0^T \kappa(q^a_s)  dN^a_s - \int_0^T \kappa(q^b_s)  dN^b_s \bigg |  \calF_t \bigg].
\end{equation}
for some function $\kappa$. In this expression, if $t$ is a jump time of $N^a$, then $\kappa(q^a_t)$ can be interpreted as the information content of the trade occurring at time $t$. It can also be seen as a measure of how hard it is to pass through the current limit order book and provoke a price change. 
Therefore we consider $\kappa_1$ to be decreasing, as it is more difficult to create a price change when the pile is bigger. 
\begin{remark}
Another suitable choice would be for $\kappa$ to depend on both queue levels and is given by
\begin{equation*}
\kappa(q^a_t, q^b_t) = \kappa_2 \bigg( \frac{q^a_t - q^b_t}{q^a_t + q^b_t} \bigg)
\end{equation*}
for a jump time of $N^a$ and $\kappa(q^b_s, q^a_s)$ for a jump time of $N^b$. Taking $\kappa_2$ decreasing would be supported by the fact that the sign of the next price movement can be predicted by the current imbalance of the order book \cite{burghardt2006measuring,gould2016queue,lehalle2017limit,stoikov2018micro}.     
However, this approach is more challenging theoretically and the interpretation for $\kappa$ is slightly less trivial. Therefore, in the following, we consider the price process \eqref{eq:price:lim}.
\end{remark}

Up to this point, it is unclear whether this limit is well defined. It is the case, but proving it requires some additional assumption on the functions $\lambda^L$ and $\lambda^C$. When $\kappa$ is constant, it was shown in \cite{jaisson2015market} that this limit is well-defined and is given explicitly by
\begin{equation*}
P_t 
=
P_0 + \kappa 
\int_0^t 
\xi(t-s)  d (N^{a}_s - N^{b}_s)
\end{equation*}
where $\xi(t) = 1 + ( 1 + \int_0^\infty \psi(s)  ds) \int_t^\infty \varphi(s)  ds$ and $\psi = \sum_{k\geq 1} \varphi^{*k}$ where $\varphi^{*k}$ stands for the $k$ fold convolution of $\varphi$. The proof involves technical computations around Hawkes processes but the finiteness of the limit is essentially down to the fact that the only activity of the Hawkes processes after time $t$ that influences this limit are the jumps which can be seen as offsprings of jumps happening before time $t$ in a cluster representation of Hawkes processes. The other jumps vanish because $N^{a}$ and $N^{b}$ have the same distribution. In our case the situation is more intricate because of the function $\kappa$: the jumps of the Hawkes processes $N^a$ and $N^b$ after time $t$ do still matter because they do not have the same weight. Consequently, in order to get the finiteness of the limit \eqref{eq:price:lim}, we need conditions that ensure that $(q^a, N^a)$ and $(q^b, N^b)$ have a stationary state and that the distribution of the queues converges fast enough to that stationary state after time $t$. Precise assumptions are specified in the following theorem.

\begin{theorem}
\label{thm:finiteness:price:lim}
For all integer $k > 0$, we define $\mathfrak{m}_k$ as
\begin{equation*}
\mathfrak{m}_k = \inf_{q \in \bbZ} \{ \lambda^L(q) - \lambda^L(q+k) + \lambda^C(q+k) -\lambda^C(q) \}
\end{equation*}
and assume that $\ul{\mathfrak{m}} = \inf_{k > 0} \mathfrak{m}_k > 0$. Suppose in addition that $\lambda^L$ is decreasing, $\lambda^C$ is increasing,  $\kappa$ is non-negative and bounded and that $\norm{\varphi}_{L^1} < 1$.
Then for every $t > 0$, the limit \eqref{eq:price:lim} is well defined and finite almost surely.
\end{theorem}

Theorem \ref{thm:finiteness:price:lim} is proved in Appendix \ref{proof:thm:finiteness:price:lim} and relies essentially on the fact that after time $t$, each Hawkes process $N^a$ and $N^b$ can be written as the sum of a Hawkes process with a time-varying baseline which is $\calF_t$-measurable and a Hawkes process which is $\calF_t$-independent and has the same dynamic as $N^a$ and $N^b$. The assumption $\norm{\varphi}_{L^1} < 1$ is classical and ensures that the Hawkes process is stable and is essential so that the Hawkes process with time-varying baseline has finite expectation. Similarly, the assumptions $\lambda^L$ decreasing and $\lambda^C$ is increasing ensure that the distribution of the queues admits a stationary distribution, while the assumption $\inf_{k > 0} \mathfrak{m}_k > 0$ ensures that the distribution of the queues converges exponentially fast to their stationary state after time $t$. The assumption $\kappa$ non-negative is natural and ensures the price movements are consistent with the order signs. The boundedness of $\kappa$ is required to ensure that the contribution of a single market order is bounded. 

Although Theorem \ref{thm:finiteness:price:lim} proves the finiteness of the limit \eqref{eq:price:lim}, it is of limited use in practice because it does not provide a closed-form formula to compute this limit. Using that $M^{a/b} = N^{a/b} - \Lambda^{a/b}$ is a martingale, where $\Lambda^{a/b}_t = \int_0^t \lambda^{a/b}_s  ds$, we see that \eqref{eq:price:lim} can be rewritten as
\begin{align*}
P_t = P_0 
&+ \int_0^t \kappa(q^a_s)  dN^a_s - \int_0^t \kappa(q^b_s)  dN^b_s
\\
&+
\lim_{T\to\infty} 
\E\bigg[\int_t^T \kappa(q^a_s) \lambda^a_s  ds - \int_t^T \kappa(q^b_s) \lambda^b_s  ds \bigg |  \calF_t \bigg]
\end{align*}
and therefore 
\begin{equation*}
P_t = P_0 + \int_0^t \kappa(q^a_s)  dN^a_s - \int_0^t \kappa(q^b_s)  dN^b_s
+
\int_t^\infty
\E[ \kappa(q^a_s) \lambda^a_s - \kappa(q^b_s) \lambda^b_s  |  \calF_t ]
  ds.
\end{equation*}
The computation of $\E[ \kappa(q^{a/b}_s) \lambda^{a/b}_s  |  \calF_t ]$ then becomes intricate because of the correlation between $\lambda^{a/b}$ and $q^{a/b}$ on one hand, and the non-linearity of the function $\kappa$ on the other hand.


\subsection{Market impact}
\label{sec:mi}

We now consider the insertion of a metaorder within this model using only limit orders. Without loss of generality, we consider a sell metaorder so that the limit orders are posted on the ask side and will eventually be hit by market orders.

We model the arrival times of this metaorder as the jump times of a point process $N^o$ independent of $N^a$, $N^b$, $L^a$, $L^b$, $C^a$ and $C^b$. 
In the spirit of \cite{jaisson2015market,jusselin2020noarbitrage}, we assume that $N^o$ is a non-homogeneous Poisson process with intensity $\nu$. We also assume that it is adapted to $(\calF_t)_t$, in the sense that $N^o_t$ is $\calF_t$ measurable for all $t$, and $(N^o_s - N^o_t)_{s \geq t}$ is independent of $\calF_t$. We write $\wb q^a$ the resulting process counting the orders available on the ask side. In a highly endogenous market, market orders only represent a smaller fraction of the total amount of orders. Thus, even if we consider a metaorder executing a significant volume, it would still only represent a smaller fraction of the limit orders posted on a market. Moreover, the dynamic of the limit order book is strongly influenced by the behaviour of the liquidity provider, who themselves adapt their trading strategies depending on the current limit order book state. In that case, we have
\begin{equation*}
\wb q^a_t = q_0^a + \wb L^{a}_t - \wb C^{a}_t - N^{a}_t + N^{o}_t
\end{equation*}
where $\wb L^{a}$ and $\wb C^{a}$ are Poisson inhomogeneous processes with stochastic intensities
\begin{equation*}
\wb \lambda_t^{L,a} = \lambda^L(\wb q_{t-}^{a})
\quad \text{ and } \quad 
\wb \lambda_t^{C,a} = \lambda^C(\wb q_{t-}^{a}).
\end{equation*}
Note that although the market digests past orders originating the metaorder, it cannot predict the arrival of future orders. Therefore, when computing the price at time $t$ as a conditional expectation of a functional of future order flows, we should truncate the metaorder at time $t$, i.e. we replace $(N^o_s)_s$ by $(N^o_{s\wedge t})_s$. To that extent, we write
\begin{equation*}
N_{s}^{o,t} = N^{o}_{s \wedge t}
\quad \quad \text{ and } \quad \quad
\wb q^{a,t}_{s} = \wb L^{a,t}_{s} - \wb C^{a,t}_{s} - N^{a}_t + N_{s}^{o,t}
\end{equation*}
where $\wb L^{a,t}$ and $\wb C^{a,t}$ are Poisson inhomogeneous process with stochastic intensities
\begin{equation*}
\wb \lambda_{s}^{L,a,t} = \lambda^L(\wb q_{s-}^{a,t})
\quad \text{ and } \quad 
\wb \lambda_{s}^{C,a,t} = \lambda^C(\wb q_{s-}^{a,t}).
\end{equation*}
Using these definitions, the price in presence of the metaorder is given by
\begin{equation}
\label{eq:price:lim:metaorder}
\wb P_t = P_0 + \lim_{T\to\infty}
\E\bigg[\int_0^T \kappa(\wb q^{a,t}_{s})  dN^a_s - \int_0^T \kappa(q^b_s)  dN^b_s \bigg |  \calF_t \bigg].
\end{equation}
This limit is well defined and can be established using similar arguments to those employed in the proof of Theorem \ref{thm:finiteness:price:lim}. For the sake of conciseness, the proof is omitted.
Combining Equations \eqref{eq:price:lim} and \eqref{eq:price:lim:metaorder}, we can define the pathwise market impact. Although most studies focus on the the average market impact $\mathbb{E}[\mathrm{MI}_t]$, previous work has shown that price impact can also be defined without averaging, see Bellani et al. \cite{bellani2023price}. In our case, it is possible to study the pathwise market impact because we can compute the price both with and without the metaorder. To do that, we need to assume that the processes $q^a$ and $\wb q^a$ are coupled, in the sense that they originate the same Poisson point measure. More precisely, we have
\begin{equation*}
L^{a}_s = \int_0^s \int_0^\infty \1_{\{z \leq \lambda^{L,a}_u\}}  \pi^{L,a}(du,dz)
\quad \text{ and } \quad
C^{a}_s = \int_0^s \int_0^\infty \1_{\{z \leq \lambda^{C,a}_u\}}  \pi^{C,a}(du,dz)
\end{equation*}
by \eqref{eq:def:pointmeasure} and then we take
\begin{equation*}
  \wb L^{a,t}_s = \int_0^s \int_0^\infty \1_{\{z \leq \wb \lambda^{L,a}_u\}}  \pi^{L,a}(du,dz)
  \quad \text{ and } \quad
  \wb C^{a,t}_s = \int_0^s \int_0^\infty \1_{\{z \leq \wb \lambda^{C,a}_u\}}  \pi^{C,a}(du,dz).
\end{equation*}
 Using these definitions, we are ready to define to pathwise market impact.
\begin{definition}
The pathwise \textbf{market impact} $\mathrm{MI}_t$ of the metaorder at time $t$ is given by
\begin{equation*}
\mathrm{MI}_t = \wb P_t - P_t.
\end{equation*}
\end{definition}

Unlike the price of Section \ref{sec:model}, this expression can be simplified 
\begin{align*}
\mathrm{MI}_t 
&= 
\lim_{T\to\infty}
\E\bigg[\int_0^T \kappa(\wb q^{a,t}_{s})  dN^a_s - \int_0^T \kappa(q^b_s)  dN^b_s \bigg |  \calF_t \bigg]
\\
&\quad\quad\quad\quad-
\lim_{T\to\infty}
\E\bigg[\int_0^T \kappa(q^a_{s})  dN^a_s - \int_0^T \kappa(q^b_s)  dN^b_s \bigg |  \calF_t \bigg]
\\
&= 
\lim_{T\to\infty}
\E\bigg[\int_0^T \big( \kappa(\wb q^{a,t}_{s}) - \kappa(q^a_s) \big)  dN^a_s \bigg |  \calF_t \bigg].
\end{align*}
This limit is well defined by Theorem \ref{thm:finiteness:price:lim}. Moreover, following the same lines as Lemma \ref{lem:control:diff:q} in Appendix \ref{proof:thm:finiteness:price:lim}, we can show that for all $s \geq 0$, $\wb q^{a,t}_s \geq q^a_s$. Therefore, when $\kappa$ is decreasing, we have $\kappa(\wb q^a_s) \leq \kappa(q^a_s) $. This ensures that
\begin{equation*}
    \int_0^T \big( \kappa(\wb q^{a,t}_s) - \kappa(q^a_s) \big)  dN^a_s \to \int_0^\infty \big( \kappa(\wb q^{a,t}_s) - \kappa(q^a_s) \big)  dN^a_s
\end{equation*}
almost surely as $T \to \infty$ and thus, using the monotone convergence theorem, we obtain
\begin{equation}
\label{eq:expression:mi}
\mathrm{MI}_t 
=
\E\bigg[\int_0^\infty \big( \kappa(\wb q^{a,t}_s) - \kappa(q^a_s) \big)  dN^a_s \bigg |  \calF_t \bigg].
\end{equation}

\label{sec:scaling}
\subsection{Scaling limit of the limit order book model}

Following the same ideas as in \cite{jusselin2020noarbitrage}, we would like to build a scaling limit of the market impact to get a better understanding of how it behaves on large time scales. 

We start by computing the scaling limit of the limit order book in the model described in Section \ref{sec:model}. Let $T > 0$ be the final time horizon, representing the end of the metaorder. We consider the same model as in Section \ref{sec:model}, with the additional exponent $T$. Specifically, we consider two independent Hawkes processes $N^{a,T}$ and $N^{b,T}$ with a common baseline intensity $\mu_T \geq 0$ and self-exciting kernel $\varphi^T$ so that their intensities are given by
\begin{equation*}
\lambda^{a,T}_t = \mu_T + \int_0^{t-} \varphi^T(t-s)  dN^{a, T}_s
\;\;\text{ and }\;\;
\lambda^{b,T}_t = \mu_T + \int_0^{t-} \varphi^T(t-s)  dN^{b, T}_s.
\end{equation*}
In many empirical studies (e.g., \cite{bacry2015hawkes}), the Hawkes kernel is observed to decay approximately as a power law with an exponent just above unity. As a concrete example, one may take
$\varphi^T(t) = (t + \tau_0)^{-1.1}$,
where $\tau_0 > 0$ is a short-time regularization constant. This choice transparently exhibits the slowly decaying, long-memory behavior seen in high-frequency order-flow data.

We then consider the limit order book and suppose that the volumes $q^{b,T}$ and $q^{a,T}$ satisfy
\begin{equation*}
q^{a,T}_t = L^{a,T}_t - C^{a,T}_t - N^{a,T}_t \quad \text{ and } \quad q^{b,T}_t = L^{b,T}_t - C^{b,T}_t - N^{b,T}_t.
\end{equation*}
Furthermore, we assume that $L^{a}$ and $L^{b}$ are two time-inhomogeneous Poisson processes whose intensities are given by 
\begin{equation*}
\lambda_t^{L,a,T} = \lambda^{L,T}(q_{t-}^{a,T})
\quad \text{ and } \quad 
\lambda_t^{L,b,T} = \lambda^{L,T}(q_{t-}^{b,T})
\end{equation*}
for some function $\lambda^{L,T}$; and that $C^{a,T}$ and $C^{b,T}$ are two time-inhomogeneous Poisson processes whose intensities are given by 
\begin{equation*}
\lambda_t^{C,a,T} = \lambda^{C,T}(q_{t-}^{a,T})
\quad \text{ and } \quad 
\lambda_t^{C,b,T} = \lambda^{C,T}(q_{t-}^{b,T})
\end{equation*}
for some function $\lambda^{C,T}$. Concrete examples of such intensities have been empirically calibrated in \cite{huang2015simulating}.

The scaling limit of the Hawkes processes driving the market orders has already been extensively studied, see for instance \cite{jaisson2015limit,jaisson2016rough,jusselin2020noarbitrage,durin2023two} and we recall here some results from \cite{jaisson2016rough}.

\begin{assumption}
\label{assumption:jaisson}
There exists a function $\varphi$ such that $\varphi^T = a_T \varphi$ for some sequence $a_T \to 1$ and $\norm{\varphi}_{L^1} = 1$. Moreover, there exists $0 < \alpha < 1$ such that the limits
\begin{equation*}
K = \lim\limits_{t \to \infty} t^\alpha \int_t^\infty \varphi(s) ds, 
\;\;
\lambda = \lim\limits_{T \to \infty} (1-\alpha) K^{-1} T^\alpha (1-a_T)
, 
\;\;
\text{ and }
\mu^* = \lim\limits_{T \to \infty} T^{1-\alpha} \mu_T 
\end{equation*}
are finite.
\end{assumption}

Under this assumption, \cite{jaisson2015market} proves that the long-term average intensity of the Hawkes processes $N^{a, T}$ and $N^{b, T}$ is $\beta_T = (1-a_T)^{-1}\mu_T$ and therefore the average number of trades from $N^{a,T}$ and $N^{b,T}$ on $[0,T]$ scales as $T \beta_T$. Thus it is natural to rescale each Hawkes process by $T\beta_T$. We define for $x \in \{a,b\}$ the sequences
\begin{equation*}
\wh {N}^{x, T}_t = \frac{1}{T\beta_T} N^{x, T}_{tT}\text{, }
\quad \quad
\wh {\Lambda}^{x, T}_t = \frac{1}{T\beta_T} \int_0^{tT} \lambda^{x, T}_s ds
\end{equation*}
and
\begin{equation*}
    \quad \wh M^{x, T}_t = \sqrt{T\beta_T} (\wh {N}^{x, T}_t - \wh {\Lambda}^{x, T}_t) = \frac{1}{\sqrt{T\beta_T}} M^{x, T}_{tT}.
\end{equation*}
\begin{proposition}
\label{prop:convergence:jaisson}
Suppose that Assumption \ref{assumption:jaisson} holds. Then, for $x \in \{a, b\}$, we have 
\begin{equation*}
\wh N^{x,T} \to X^{x}
\;\;\text{and}\;\;
\wh M^{x,T} \to Z^{x}
\end{equation*}
in distribution, for the Skorokhod topology on compact subsets of $[0, \infty)$, 
where $X^{x}$ is an increasing process with derivative $Y^{x}$ satisfying
\begin{equation}
\label{eq:def:Y}
Y^{x}_t 
=
F^{\alpha, \lambda}(t)+\frac{1}{\sqrt{\mu^* \lambda}} \int_0^t f^{\alpha, \lambda}(t-s) \sqrt{Y_s^{x}} d B_s
\end{equation} 
for some Brownian motion $B$ and $Z^x$ is a continous martingale with quadratic variation $X^x$. Here $f^{\alpha, \lambda}$ is defined by
\begin{equation*}
f^{\alpha,\lambda}(x) = \lambda x^{\alpha-1} E_{\alpha, \alpha}(-\lambda x^\alpha)
\end{equation*}
where $E_{\alpha, \beta}$ is the $(\alpha, \beta)$-Mittag-Leffler function in Haubold et al. \cite{haubold2011mittag}, that is
\begin{equation*}
E_{\alpha, \beta}(x) = \sum_{k=0}^\infty \frac{x^k}{\Gamma(\alpha k + \beta)}
\end{equation*}
and $F^{\alpha, \lambda}$ is defined by
\begin{equation*}
F^{\alpha,\lambda}(x) = \int_0^x f^{\alpha,\lambda}(y)  dy.
\end{equation*}

\end{proposition}

We are now ready to study the limit order book $(q^{a,T}, q^{b,T})$. It is natural to look for limits where the quantities available on each queue scale like the average of trades. Therefore, we define 
\begin{equation*}
\wh {q}^{a, T}_t = \frac{1}{T\beta_T} q^{a, T}_{tT}
\quad \text{ and } \quad
\wh {q}^{b, T}_t = \frac{1}{T\beta_T} q^{b, T}_{tT}.
\end{equation*}
where the scaling $(T\beta_T)^{-1}$ ensures these limits are non-degenerate. Similarly, we define
\begin{equation*}
\wh {L}^{a, T}_t = \frac{1}{T\beta_T} L^{a, T}_{tT}
\quad \text{ and } \quad
\wh {L}^{b, T}_t = \frac{1}{T\beta_T} L^{b, T}_{tT}
\end{equation*}
and 
\begin{equation*}
\wh {C}^{a, T}_t = \frac{1}{T\beta_T} C^{a, T}_{tT}
\quad \text{ and } \quad
\wh {C}^{b, T}_t = \frac{1}{T\beta_T} C^{b, T}_{tT}.
\end{equation*}
We also normalise the intensities by a factor $\beta_T$ to ensure consistencies in the notations when considering the compensator of $(\wh {L}^{a, T}, \wh {L}^{b, T}, \wh {C}^{a, T}, \wh {C}^{b, T})$ and we write
\begin{align*}
&\wh {\lambda}^{a, L, T}_t = \frac{1}{\beta_T} \lambda^{a, L, T}_{tT}
,\quad
\wh {\lambda}^{b,L, T}_t = \frac{1}{\beta_T} \lambda^{b,L, T}_{tT}
,
\\
&\wh {\lambda}^{a,C, T}_t = \frac{1}{\beta_T} \lambda^{a, C,T}_{tT}
,\;\text{ and }\;
\wh {\lambda}^{b, C,T}_t = \frac{1}{\beta_T} \lambda^{b,C, T}_{tT}.
\end{align*}
Using their definition, we see that for $x \in \{a, b\}$ and $y \in \{L,C\}$, we have
\begin{equation*}
\wh {\lambda}^{y, x, T}_t = \frac{1}{\beta_T} \lambda^{y, x, T}_{tT} = \frac{1}{\beta_T} \lambda^{y, T}({q}^{y, x, T}_{tT})
= \frac{1}{\beta_T} \lambda^{y, T}(T\beta_T \wh {q}^{y, x, T}_{t}).
\end{equation*}
Therefore, we assume the following.
\begin{assumption}
\label{assumption:scaling:lambda}
There exist $\lambda^{L}$ and $\lambda^{C}$ such that for all $q$, we have
\begin{equation*}
\lambda^{L,T}(q) = \beta_T \lambda^{L}\big(q / (T\beta_T)\big)
\quad \text{ and } \quad
\lambda^{C,T}(q) = \beta_T \lambda^{C}\big(q / (T\beta_T)\big).
\end{equation*}
Moreover, $\lambda^L$ is decreasing, $\lambda^C$ is increasing and $\lambda^{L} - \lambda^C$ is Lipschitz continuous.
\end{assumption}

In addition to Assumption \ref{assumption:scaling:lambda}, we also need to ensure the initial conditions converge to guarantee convergence of the full process $\wh q^{T,x}$. More precisely, we assume in the following the convergence of the rescaled initial state of the orderbook.
\begin{assumption}
\label{assumption:scaling:initial}
There exist $q^{a}_0$ and $q^{b}_0$ two $\calF_0$-measurable random variables such that 
\begin{equation*}
\wh q^{T,a}_0 \to q^{a}_0 \quad \quad \text{ and } \quad \quad \wh q^{T,b}_0 \to q^{b}_0
\end{equation*}
in distribution.
\end{assumption}

Note that under Assumption \ref{assumption:scaling:lambda}, we have 
\begin{equation*}
\wh \lambda^{y,x,T}_t = \lambda^{y}(\wh q^{y,x,T}_t).
\end{equation*}
This implies the following proposition, proved in Appendix \ref{sec:proof:prop:convergence:LOB}.
\begin{proposition}
\label{prop:convergence:LOB}
Let $I$ be a closed interval of $[0, \infty)$ and suppose that Assumptions \ref{assumption:jaisson}, \ref{assumption:scaling:lambda} and \ref{assumption:scaling:initial} hold. Then, for $x\in\{a,b\}$, the process $(\wh q^{T, x}, \wh L^{T, x}, \wh C^{T, x})$ converges in distribution towards the continuous process $(q^{x}, L^{x}, C^{x})$ for the Skorokhod topology on $I$. Moreover, we have $q^{x} = q^{x}_0 + L^{x} - C^{x} - X^{x}$ where $X^{x}$ is defined in Proposition \ref{prop:convergence:jaisson} and where
\begin{equation*}
\begin{cases}
L^{x} = \int_0^t \lambda^{L}(q_s^x) ds,
\\
C^{x} = \int_0^t \lambda^{C}(q_s^x)ds.
\end{cases}
\end{equation*}
\end{proposition}

Note that although Proposition \ref{prop:convergence:LOB} is proved for a compact interval $I$, the processes $q^{x}, L^{x}, C^{x}$ and $X^{x}$ can be defined on the whole half-line $[0, \infty)$ with the same definitions. Moreover, the process $q^x$ is differentiable with derivative
\begin{equation}
\label{eq:ODE:q}
	q^{x\prime}_t = \lambda^{L}(q^{x}_t) - \lambda^{C}(q^{x}_t) - Y^{x}_t.
\end{equation}

\begin{corollary}
\label{cor:linearsolving}
The behaviour of the limiting order book is deterministic conditional on the market volatility. Moreover, when $\lambda^L(x) - \lambda^C(x) = c_\lambda x + d_\lambda$ for all $x$, the ordinary differential equation \eqref{eq:ODE:q} can be solved explicitly and we have
\begin{equation*}
q^{a/b}_t = 
q^{a/b}_0 e^{c_\lambda t} + \int_0^t e^{c_\lambda (t-s)} (d_\lambda-Y^{a/b}_s)  ds.
\end{equation*}  
\end{corollary}

\subsection{Scaling limit of the market impact model}

We are now ready to study the scaling limit of the market impact. We follow the approach initiated in \cite{jusselin2020noarbitrage,durin2023two}. However, we do not have a closed-form expression for the price and for the market impact, making the computations more intricate. In particular, dealing with the conditional expectation in the definition of the market impact requires delicate arguments. Therefore, instead of studying the pathwise market impact, we only study the average market impact.

We consider here as well a family of metaorders $N^{o,T}$ passed through limit orders on the ask side. We assume that $N^{o,T}$ is a non-homogeneous Poisson process with intensity $\nu^T$ and we write $\wb q^{a,T}$ the resulting process counting the orders available on the ask side. Following \cite{jusselin2020noarbitrage,durin2023two}, we assume that the size of a metaorder is measured relatively to the total market orders volume, which is of order $T\beta_T$ on $[0,T]$.
\begin{assumption}
\label{assumption:f}
There exists a function $f:[0,\infty) \mapsto [0,\infty)$ satisfying $f(t) = 0$ for $t > 1$ such that for any $t \geq 0$, we have
\begin{equation*}
\nu^T(t) = \beta_T f(t/T).
\end{equation*}
\end{assumption}

Following the same proof as Proposition \ref{prop:convergence:LOB}, we can show that the rescaled aggregated queue in presence of the metaorder also converges in distribution for the Skorokhod topology. More precisely, we define
\begin{equation*}
\check{q}^{a,T}_s = (T\beta_T)^{-1} \wb q^{a,T}_{tT}.
\end{equation*}
The convergence of $\check q^{a,T}$ can be shown using the same proof as Proposition \ref{prop:convergence:LOB} and is expressed in the following result.
\begin{proposition}
\label{prop:convergence:LOBMO}
Let $I$ be a closed interval of $[0, \infty)$ and suppose that Assumptions \ref{assumption:jaisson} and \ref{assumption:scaling:lambda} hold. Then, $\check q^{a,T}$ converges in distribution for the Skorokhod topology on $I$ towards the continuous process $\wb q^{a}$. Moreover, we have $\wb q^{a} = q^{a}_0 + \wb L^{a} - \wb C^{a} - \wb X^{a} + F$ where $X^{a}$ is defined in Proposition \ref{prop:convergence:jaisson}, $F(t) = \int_0^t f(s) ds$ and
\begin{equation*}
\begin{cases}
\wb L^{a} = \int_0^t \lambda^{L}(\wb q_s^a) ds,
\\
\wb C^{a} = \int_0^t \lambda^{C}(\wb q_s^a)ds.
\end{cases}
\end{equation*}
\end{proposition}

Similarly to Proposition \ref{prop:convergence:LOB}, we the processes $\wb q^{a}, \wb L^{a}$ and $ \wb C^{a}$ can be defined on the whole half-line $[0, \infty)$ with the same expressions. Moreover, $\wb q^{a}$ is differentiable with derivative 
\begin{equation}
\label{eq:ODE:qmi}
	\wb q^{a\prime}_t = \lambda^{L}(\wb q^{a}_t) - \lambda^{C}(\wb q^{a}_t) - Y^{a}_t + f(t).
\end{equation}
\begin{corollary}
\label{cor:linearsolvingwithmetaorder}
The behaviour of the limiting order book $\wb q^{a}$ in the presence of the metaorder is deterministic conditional on the market volatility. When $\lambda^L - \lambda^C$ is an affine function, say $\lambda^L(x) - \lambda^C(x) = c_\lambda x + d_\lambda$ for all $x$, the ordinary differential equation \eqref{eq:ODE:qmi} can be solved explicitly and we have
\begin{equation*}
\wb q^{a}_t = 
q^{a}_0 e^{c_\lambda t}
+
\int_0^t e^{c_\lambda (t-s)} \big(d_\lambda - Y^{a}_u + f(u)\big)  du.
\end{equation*}  
\end{corollary}
Moreover, since $f$ is positive, we see that $\wb q^{a}_t \geq q^{a}_t$ for all $t$. Furthermore, since $\lambda^{L} - \lambda^{C}$ is decreasing, we also have
\begin{equation*}
|\wb q^{a\prime}_t - q^{a\prime}_t| \leq \int_0^t f(s)  ds.
\end{equation*}
Note also from \eqref{eq:ODE:q} and \eqref{eq:ODE:qmi} that $q^{a}$ and $\wb q^{a}$ both solve an ordinary differential equation of the form $y' = u(y) + v$ where $u = \lambda^{L} - \lambda^{C}$ and $v$ is given by $v = - Y^{a}$ for $ q^{a}$ and $v = - Y^{a} + f$ for $\wb q^{a}$. Thus the behaviour of $|\wb q^{a\prime}_t - q^{a\prime}_t|$ is directly linked to $\lambda^{L} - \lambda^{C}$. We assume the following.
\begin{assumption}
\label{assumption:mean_rev}
There exist $c > 0$ such that for all $q$ and all $x$, we have
\begin{equation*}
\lambda^{L}(q) - \lambda^{L}(q + x) + \lambda^{C}(q + x) - \lambda^{C}(q) \geq c x.
\end{equation*}
\end{assumption}

We write $\mathrm{MI}^T$ the market impact in this model, which is associated to a function $\kappa^T$ replacing the function $\kappa$ used in Section \ref{fig:limit:shape}. Again, it is natural that the market impact scales at the same rate as the order flow and therefore, we define the rescaled market impact $\wh{\mathrm{MI}}_t^T$ by
\begin{equation*}
\wh{\mathrm{MI}}_t^T = \frac{1}{T\beta_T} \mathrm{MI}^T_{tT}. 
\end{equation*}
Using \eqref{eq:expression:mi}, it can be rewritten as
\begin{equation*}
\wh{\mathrm{MI}}_t^T = \frac{1}{T\beta_T} \E\bigg[\int_0^\infty \big( \kappa^T(\wb q^{a,T,t}_s) - \kappa^T(q^{a,T}_s) \big)  d N^{a,T}_s \bigg |  \calF_{tT}^T \bigg]
\end{equation*}
where $\wb q^{a,T,t}$ is defined from $\wb q^{a,T}$ by truncating the limit orders to times before $t$ only (see Section \ref{sec:mi} for details). Writing $\check q^{a,T,t}_s = (T\beta_T)^{-1} \wb q^{a,T,t}_{tT}$ and $\wh \calF_{t}^T = \calF_{tT}^T$, we have
\begin{equation*}
\wh{\mathrm{MI}}_t^T = \E\bigg[\int_0^\infty \big( \kappa^T(T\beta_T \check q^{a,T,t}_s) - \kappa^T(T\beta_T \wh q^{a,T}_s) \big)  d \wh N^{a,T}_s \bigg |  \wh \calF_{t}^T \bigg].
\end{equation*}
Therefore, it is natural to assume the following 
\begin{assumption}
\label{assumption:kappa}
There exists a Lipschitz continuous function $\kappa$ such that for all $T$ and all $x$, we have $ \kappa^T(x) = \kappa(x / (T\beta_T))$.
\end{assumption}
This assumption ensures that we have
\begin{equation*}
\wh{\mathrm{MI}}_t^T = \E\bigg[\int_0^\infty \big( \kappa(\check q^{a,T,t}_s) - \kappa(\wh q^{a,T}_s) \big)  d \wh N^{a,T}_s \bigg |  \wh \calF_{t}^T \bigg].
\end{equation*}
Following previous results about the convergence of the queuing processes, it seems natural that $\wh{\mathrm{MI}}^T$ converges towards $MI$, defined for all $t \geq 0$ by
\begin{equation}
\label{eq:scaling_limit_market_impact}
\mathrm{MI}_t = 
 \E\bigg[\int_0^\infty \big( \kappa(\wb q^{a,t}_s) - \kappa(q^{a}_s) \big) Y^{a}_s  ds \bigg |  \calF_{t} \bigg]
\end{equation}
where $Y^a_s$ refers to the volatility of the market price defined in Proposition \ref{prop:convergence:jaisson}, $f$ is the rescaled metaorder intensity defined in Assumption \ref{assumption:f} and $\wb q^{a,t}$ is defined by $\wb q^{a,t}_0 = q^{a}_0$ and for all $s \geq 0$
\begin{equation*}
\wb q^{a,t\prime}_s = \lambda^{L}(\wb q^{a,t}_s) - \lambda^{C}(\wb q^{a,t}_s) - Y^{a}_s + f(s)\1_{s\leq t}.
\end{equation*}
However, the limit $\wh{\mathrm{MI}}_t^T \to \mathrm{MI}_t$ is hard to prove because in previous results, the convergence of each process is studied on compact subsets of $[0, \infty)$, and a convergence $\wh{\mathrm{MI}}_t^T \to \mathrm{MI}_t$ would require convergence on the whole set $[0, \infty)$. Nevertheless, we can still prove the weaker convergence $\E[\wh{\mathrm{MI}}_t^T] \to \E[\mathrm{MI}_t]$. The idea of the proof lies in splitting the integral $\int_0^\infty$ into $\int_0^A + \int_A^\infty$ for $A$ large enough so that the contribution of $\int_A^\infty$ remains small, while previous results ensure the convergence of the integral $\int_0^A$. This result is stated formally in the following theorem and is proved in Appendix \ref{sec:proof:thm:scalingmi}.

\begin{theorem}
\label{thm:scalingmi}
Suppose that Assumptions \ref{assumption:jaisson}, \ref{assumption:scaling:lambda}, \ref{assumption:scaling:initial}, \ref{assumption:f}, \ref{assumption:mean_rev} and \ref{assumption:kappa} hold. Then for all $t \geq 0$, $\E[\wh{\mathrm{MI}}^{T}_t]$ converges towards $\E[\mathrm{MI}_t]$, where $\mathrm{MI}_t$ is given in Equation \eqref{eq:scaling_limit_market_impact}.

\end{theorem}

\section{Applications and approximations}
\label{sec:applications}

\subsection{Asymptotic shape of the market impact}
\label{sec:shape}

Using Theorem \ref{thm:scalingmi}, we can guess the asymptotic shape of the market impact in terms of the participation rate. Denote $D = \lambda^L - \lambda^C$ and suppose that this function is invertible. Suppose also that $Y^x$ is in a stationary state, and for simplicity, that it is constant and equals $m$. Then, using \eqref{eq:ODE:q}, the stationary state of $q^x$ is given by $q^* = D^{-1}(m)$.

Suppose now that $f(t) = \gamma$ is constant on $[0,1]$ so that $\gamma$ is a proxy for the participation rate. Then, the stationary state for the queue in presence of this metaorder is given by $\wb q^* = D^{-1}(m - \gamma)$. Combining these and Theorem \ref{thm:scalingmi}, we see that the market impact is proportional to 
\begin{equation*}
\kappa(\wb q^*) - \kappa(q^*) = \kappa\big(D^{-1}(m - \gamma)\big) - \kappa\big(D^{-1}(m)\big).
\end{equation*}
Few remarks about this result. First, when $\kappa$ is constant, the market impact vanishes, which is consistent with previous studies \cite{jusselin2020noarbitrage,durin2023two}. When both $\kappa$ and $D$ are linear, so is $D^{-1}$ and therefore 
\begin{equation*}
\kappa(\wb q^*) - \kappa(q^*) = c_1 \gamma + c_2
\end{equation*}
for some constants $c_1, c_2$. This case extends to the case where $\kappa$ and $D^{-1}$ are differentiable. In that case, we have 
\begin{equation*}
\kappa(\wb q^*) - \kappa(q^*) \approx - (\kappa \circ D^{-1})'(m) \gamma
\end{equation*}
when $\gamma \to 0$. The behaviour when $\gamma \to \infty$ relies on some additional assumptions on $\kappa$ and $D$. More precisely, we assume that $D(q) \sim -cq$ when $q \to \infty$, which corresponds to the fact that there is a linear restoring force bringing the queue back to its equilibrium value. Moreover, the function $\kappa$ measures the resistance to price change. Following the ideas of dynamic theory of market liquidity in Toth et al. \cite{toth2011anomalous}, Mastromatteo et al. \cite{mastromatteo2014agent}, Donier et al. \cite{donier2015fully}, and Benzaquen and Bouchaud \cite{benzaquen2018market}, the available liquidity profile exhibits a ’V’-shaped pattern, diminishing in the vicinity of the current price while linearly increasing as one moves away from it. In this theory, it is assumed that the instantaneous density of resting orders is smallest right at the best bid and ask, and then grows approximately linearly as one moves away in price space on either side. Therefore, the total volume required to move the price of $x$ ticks should increase as $x^2$, once $x$ is large enough. In our case, it is natural to assume that $\kappa$ is linked to the inverse of this function: moving the liquidity by $\alpha$ should create a price change of $\sqrt{\alpha}$, and therefore we need to take $\kappa(q - \alpha) \approx \kappa(q) - \sqrt{\alpha}$. Assuming that this hold, and using that $D(q) \sim -cq$ when $q \to \infty$, we obtain
\begin{equation*}
\kappa(\wb q^*) - \kappa(q^*) \approx - \sqrt{c\gamma}
\end{equation*}
when $\gamma \to \infty$.

To illustrate this fact, we consider the case where $\kappa(q) = c_1 \sqrt{\log(e^{-c_2 q} + 1)}$ for some constants $c_1$ and $c_2$. This choice ensures that $\kappa(q)$ is well defined for all values of $q$ and that $\kappa(q) \sim c_3 q^{1/2}$ when $q \to \infty$. We fix $c_1 = 0.01$ and $c_2 = 1000$ for this example. We also take $\lambda^L- \lambda^C$ affine and more specifically $\lambda^L(q) - \lambda^C(q) = 0.025 - q$.
Simulating then $10000$ times the scaling limit of the market impact and averaging over all trajectories, we retrieve a power law behaviour close to a square-root behaviour when plotting $\E[\mathrm{MI}_t]$, see Figure \ref{fig:MI_fit_sqrt}.

\begin{figure}[htbp] 
  \centering
  \includegraphics[width=0.8\textwidth]{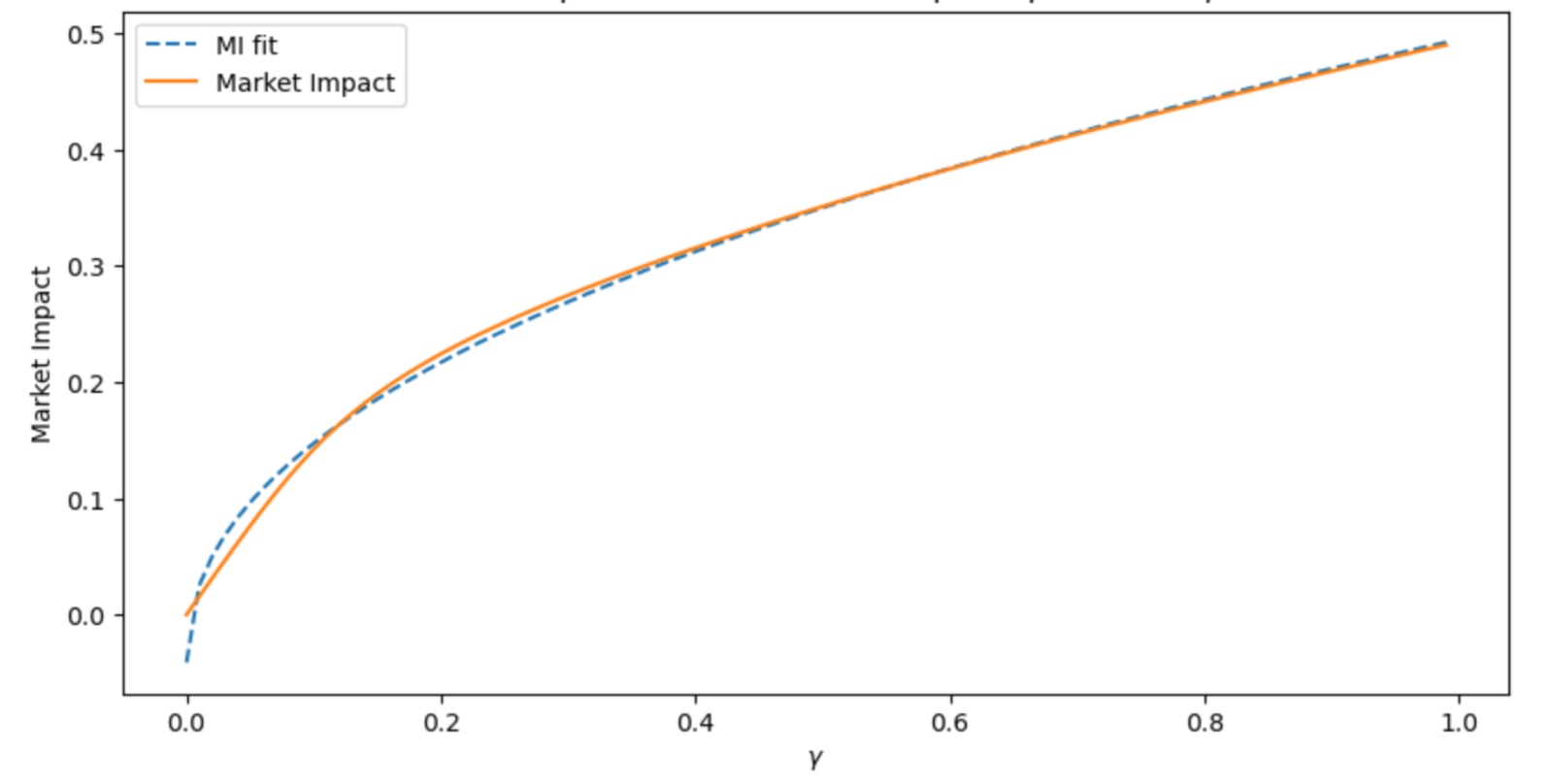}
  \caption{Power law fit of the simulated market impact with exponent 0.54.}
  \label{fig:MI_fit_sqrt}
\end{figure}

\subsection{First order approximation of the market impact}
\label{sec:first_order}

Although the results of Section \ref{sec:all_models} give a precise exact description of the market impact of a limit order, the results are quantitatively challenging to analyse, as the conditional expectation appearing in Equation \eqref{eq:expression:mi} cannot be computed explicitly.
This difficulty arises from the correlation between $\kappa(\wb q^{a,t}_s) - \kappa(q^a_s)$ and $N^a_s$. However, due to the mean-reverting properties of the jump processes $L^{a}$, $C^{a}$, $\wb L^{a}$ and $\wb C^{a}$, which are at the heart of the proof of Theorem \ref{thm:finiteness:price:lim}, both $q^a$ and $\wb q^a$ are likely to remain within a compact set with high probability. To simplify, consider a linear approximation of $\kappa$ and $\lambda^L - \lambda^C$ on this compact set in the form
\begin{equation*}
\kappa(x) \approx c_\kappa x + d_\kappa
\quad \text{ and } \quad
\lambda^L(x) - \lambda^C(x) \approx c_\lambda x + d_\lambda.
\end{equation*}
By substituting these linear approximations into \eqref{eq:expression:mi}, we obtain an explicit approximation for the market impact that is easier to study. Recall first that the dynamics of $q^a$ and $\wb q^{a,t}$ are given by
\begin{equation*}
q^a_s = q_0^a + L^{a}_s - C^{a}_s - N^{a}_s,
\quad \text{ and } \quad
\wb q^{a,t}_s = q_0^a + \wb L^{a,t}_s - \wb C^{a, t}_s - N^{a}_s + N^{o}_{s\wedge t}.
\end{equation*}
Moreover, we have
\begin{equation*}
\kappa(\wb q^{a,t}_s) - \kappa(q^a_s)
\approx
c_\kappa(\wb q^{a,t}_s - q^a_s) = c_\kappa (\wb L^{a,t}_s - \wb C^{a,t}_s - L^a_s + C^a_s + N^{o}_{s\wedge t}).
\end{equation*}
Following the same lines as in the proof of Lemma \ref{lem:expvar}, we see that for $s \geq t$, we can replace $\wb q^{a,t}_s - q^a_s$ by 
\begin{equation*}
\wb q^{a,t}_s - q^a_s 
= 
\wb q^{a,t}_t - q^a_t
-
\int_t^s \int_0^\infty 
\1_{\{
	z \leq 
	\lambda^{L}(q^a_{u-}) - \lambda^{L}(\wb q^{a,t}_{u-}) + \lambda^{C}(\wb q^{a,t}_{u-}) - \lambda^{C}(q^a_{u-} ) 
\}}
\wt \pi(du, dz)
\end{equation*}
for some random Poisson point measure $\wt \pi$ which is independent of $(\pi^{N,a}, \pi^{N,b})$. By construction, we can always take $\wt \pi$ adapted to $(\calF_t)_t$, in the sense that its restriction to $[0,t] \times [0,\infty)$ is $\calF_t$-measurable and its restriction to $(t, \infty) \times [0,\infty)$ is independent of $\calF_t$. Using the approximation $\lambda^L(x) - \lambda^C(x) \approx c_\lambda x + d_\lambda$, we observe that $\wb q^{a,t} - q^a \approx U$ where for $s\leq t$, we set $U_s = \wb q^{a,t}_s - q^a_s$ and for $s \geq t$, we define
\begin{equation*}
U_s = U_t 
-
\int_t^s \int_0^\infty 
\1_{\{
	z \leq - c_\lambda U_{u-} 
\}}
\wt \pi(du,dz).
\end{equation*}
This implies in particular that
\begin{equation*}
\E[U_s|\calF_t] = U_t 
+ c_\lambda
\int_t^s 
 \E[U_{u-} | \calF_t]
du
\quad
\text{ and hence }\quad
\E[U_s|\calF_t] = U_t 
\exp\big(c_\lambda(s-t)\big).
\end{equation*}
Thus, we obtain
\begin{align*}
\mathrm{MI}_t 
&\approx \E \bigg[ \int_0^\infty c_\kappa U_s  dN^a_s  \bigg|  \calF_t \bigg]
\\
&\approx 
\int_0^t c_\kappa U_s  dN^a_s
+
\int_0^\infty
c_\kappa 
\E[  U_s \lambda^a_s  |  \calF_t ]
  ds.
\end{align*}
For $s \geq t$, we know that that $U_s$ and $N^a$ are independent conditional on $\calF_t$ and thus we get
\begin{equation*}
\mathrm{MI}_t \approx \int_0^t c_\kappa U_s  dN^a_s + \int_t^\infty c_\kappa U_t 
\exp\big(c_\lambda(s-t)\big) \E[\lambda^a_s |  \calF_t ]  ds.
\end{equation*}
Additionally, using Lemma \ref{lemma:expressions_M}, we have 
\begin{equation*}
\E[\lambda^a_s |  \calF_t ] = 
\mu + \mu \int_0^s \psi(u)  du + \int_0^{t} \psi(s-u)  dM^{a}_u
\end{equation*}
for all $s > t$, where $M^{a} = N^{a}-\Lambda^{a}$ and $\Lambda^{a}_t = \int_0^t \lambda^{a}_u  du$.

In particular, this formula allows us to compute the impact of a single limit order at its posting time. Consider for instance $N^o$ made of a single jump at time $t_0$ and suppose that $t_0$ is large enough so that $\int_0^{s+t_0} \psi(u)  du \approx \norm{\psi}_{L^1}$ for all $s\geq 0$ This ensures that the intensity of all the processes involved in the micro-structure model is close to its stationary state. Then, for $t \leq t_0$, it is clear that $\mathrm{MI}_t = 0$. For $t = t_0$, we have $U_{t_0} = 1$ and thus
\begin{equation*}
\mathrm{MI}_{t_0} \approx \int_{t_0}^\infty c_\kappa
\exp\big(c_\lambda(s-t_0)\big) \bigg \{ \mu + \int_0^s \psi(u)  du  \mu + \int_0^{t_0} \psi(s-u)  dM^{a}_u \bigg\}  ds.
\end{equation*}
Suppose that $c_\lambda < 0$, which is a natural assumption since this ensures the aggregated queue is an ergodic process. Using that $\int_0^{s+t_0} \psi(u)  du \approx \norm{\psi}_{L^1} = (1-\norm{\varphi}_{L^1})^{-1} \norm{\varphi}_{L^1}$ for all $s\geq 0$, 
we can simplify this expression and we obtain a formula for the instantaneous market impact of a limit order as follows
\begin{equation*}
\mathrm{MI}_{t_0} 
\approx 
	- \frac{c_\kappa}{c_\lambda}
  \frac{\mu}{1-\norm{\varphi}_{L^1}}
+
	\int_0^\infty c_\kappa \exp( c_\lambda s) 
	\int_0^{t_0} \psi(s+t_0-u)  dM^{a}_u 
	 ds.
\end{equation*}
Then, on average, the instantaneous market impact of a limit order with unitary volume is approximately
\begin{equation*}
  \pm \frac{c_\kappa}{c_\lambda}
  \frac{\mu}{1-\norm{\varphi}_{L^1}}.
\end{equation*}
The terms in this expression can be interpreted as follows. The market impact is proportional to $c_\kappa / c_\lambda$, which represents the ratio of the sensitivity of price stickiness to volume, $c_\kappa = \kappa'(q)$, and the sensitivity of $\lambda^L - \lambda^C$ to volume, reflecting the mixing properties of the aggregated queues. The second term, $(1-\norm{\varphi}_{L^1})^{-1} \mu$, is classical and denotes the average long-term intensity of the market order flow.

The same ideas can also be used to derive an approximation of the scaling limit of the market impact when $\lambda^L - \lambda^C$ and $\kappa$ are both assumed linear. Specifically, we still assume $
\kappa(x) = c_\kappa x + d_\kappa$ and $\lambda^L(x) - \lambda^C(x) = c_\lambda x + d_\lambda$ and we fix $t \geq 0$. Corollaries \ref{cor:linearsolving} and \ref{cor:linearsolvingwithmetaorder} ensure we can solve exactly the ordinary differential equation defining the scaling limits $\wb q^{a,t}$ and $q^{a}$ of the queues. In this case, we have
\begin{equation*}
\begin{split}
q^{a}_s &= 
q^{a}_0 e^{c_\lambda s}
+
\int_0^s e^{c_\lambda (s-u)} (d_\lambda - Y^{a}_u)  du,
\\
\wb q^{a, t}_s &= 
q^{a}_0 e^{c_\lambda s}
+
\int_0^s e^{c_\lambda (s-u)} \big(d_\lambda - Y^{a}_u + f(u)\1_{\{u\leq t\}}\big)  du.
\end{split}
\end{equation*}

Therefore, the difference
\begin{equation*}
\kappa(\wb q^{a, t}_s) - \kappa(q^{a}_s)
=
c_\kappa
\int_0^{s \wedge t} e^{c_\lambda (s-u)} f(u)  du
\end{equation*}
is deterministic.
Then the scaling limit of the passive market impact of a metaorder is approximately given by
\begin{align}
\label{eq:scaling:MI}
\mathrm{MI}_t 
&\approx 
c_\kappa \int_0^t \int_0^{s} e^{c_\lambda (s-u)} f(u)  du  Y_s^{a}  ds
+ 
c_\kappa \int_t^\infty \int_0^{t} e^{c_\lambda (s-u)} f(u)  du  \E[Y_s^{a}|\calF_t]  ds.
\end{align}
This formula applies throughout the execution of a metaorder and can naturally be extended beyond the end of the metaorder by setting $f(t) = 0$ for $t$ after the metaorder ends. From a practical perspective, this expression is intuitive. The constant $c_\kappa$ reflects how the presence of additional liquidity in the limit order book influences the price. The trading flow of the strategy, represented by $f(u)$, affects how the limit order book digests this additional liquidity via $(e^{c_\lambda (s-u)} f(u))_{s \geq u}$. The terms $Y_s^{a}$ and $\E[Y_s^{a}|\calF_t]$ refer to the market price’s variance and forward variance curve, respectively, for $s \geq t$. Notably, all parameters in \eqref{eq:scaling:MI} are observable on the market.

\subsection{Evaluation of a metaorder's impact}

Using the first-order approximation of the market impact of a limit order and the formula for market impact for market orders obtained in \cite{jusselin2020noarbitrage}, a broker can estimate the market impact of his trading strategy \textit{a posteriori}. More precisely, we consider a strategy that has used an intensity $f$ of limit orders and $g$ of market orders. Suppose in addition that 
\begin{equation*}
\kappa(x) = c_\kappa x + d_\kappa
\end{equation*}
and
\begin{equation*}
\lambda^L(x) - \lambda^C(x) = c_\lambda x + d_\lambda.
\end{equation*}
We denote by $\kappa^*$ the average value of $\kappa(q_t)$ when $q_t$ is in its stationary state.

Using the formulas of \cite{jusselin2020noarbitrage} for the scaling limit of the market impact of market orders, we see that at each time, the market impact of the whole strategy is given by
\begin{equation*}
\mathrm{MI}_t = \mathrm{MI}_t^{l} + \mathrm{MI}_t^{m}
\end{equation*}
where 
\begin{equation*}
\begin{split}
\mathrm{MI}_t^{l} &= c_\kappa \int_0^t \int_0^{s} e^{c_\lambda (s-u)} f(u)  du  \sigma_s  ds
+ 
c_\kappa \int_t^\infty \int_0^{t} e^{c_\lambda (s-u)} f(u)  du  \xi_t(s)  ds
\\
\mathrm{MI}_t^{m} &= \kappa^* \int_0^t \big(1 + \lambda^{-1} (t-s)^{-\alpha}\big) g(s)  ds
\end{split}
\end{equation*}
where $\xi_t(s)$ is the forward volatility curve, \textit{i.e.} formally given by $\xi_t(s) = \E[\sigma_s|\calF_t]$.

More rigorously, when market orders and limit orders are executed simultaneously, the dynamics of the queues and their scaling limits are influenced by the presence of market orders. However, we disregard this effect, as the typical volume of market orders in an order book is negligible when compared to the volumes of limit orders and cancellations. Consequently, the inclusion of additional market orders would not significantly alter the queue dynamics, and we therefore maintain the approximation derived in Section \ref{sec:first_order}. Thus, it is important to note that, in this case, the impact of market and limit orders can be assessed independently using the previously established formulas.



Figure \ref{fig:limit:shape} presents $3$ simulations of the pathwise market impact $(\mathrm{MI}^{l}_t)_{t \geq 0}$ of limit orders. Two distinct behavioral phases can be observed:
\begin{itemize}
  \item For $t \leq 1$: The process $(\mathrm{MI}^{l}_t)_{t \leq 1}$ exhibits monotonic behavior. This reflects the intuitive result that as more limit orders are placed on the ask side (or bid side), it becomes increasingly difficult for the price to go up (or down). This exerts a counterforce on price movements, pushing the price in the opposite direction to the side with higher order volume.
  \item For $t > 1$: The behavior of $\mathrm{MI}^{l}_t$ becomes less predictable. Specifically, we express the market impact after time $1$ as 
  \begin{equation*}
    \mathrm{MI}_t^{l} = c_\kappa \int_0^t \int_0^{1} e^{c_\lambda (s-u)} f(u)  du  \sigma_s  ds
+ 
c_\kappa \int_t^\infty \int_0^{1} e^{c_\lambda (s-u)} f(u)  du  \xi_t(s)  ds.
  \end{equation*}
  This means that for all $\delta > 0$, we have
  \begin{equation*}
    \mathrm{MI}^l_{t + \delta} - \mathrm{MI}^l_{t} = \int_0^1 e^{-c_\lambda u} f(u)du \times \int_t^{t +\delta} e^{c_\lambda s}\big(\sigma_s - \xi_t(s)\big) ds
  \end{equation*}
  which means that the behavior of $\mathrm{MI}^l_t$ for $t > 1$ depends on the sign of the process $(\sigma_s - \xi_t(s))_{s\geq t > 1}$. 
\end{itemize}

Note also that the expectation of $(\sigma_s - \xi_t(s))_{s\geq t > 1}$ is always $0$. This ensures that there is no additional market impact on average after time $1$. Moreover, we have
\begin{equation*}
  \E[\mathrm{MI}_t] = 
 \E\bigg[\int_0^\infty \big( \kappa(\wb q^{a,t}_s) - \kappa(q^{a}_s) \big) Y^{a}_s  ds \bigg]
\end{equation*}
which shows again that $\E[\mathrm{MI}_t]$ is monotonic for $t \leq 1$ and constant for $t > 1$. 
This result can be interpreted as follows: in our model, placing limit orders on one side of the order book generates a resistance force against price movements. The more volume in the queue, the more difficult it is for the price to move. By the conclusion of the limit orders strategy, the queues converge towards their stationary state, with the price having already been impacted by our limit orders. The absence of any relaxation or reduction in market impact beyond $t = 1$ is explained by the fact that after the end of the execution of the metaorder, the queues quickly come back to their stationary state. In the scaling limit, this effect is even immediate.

\begin{figure}[htbp] 
  \centering
  \includegraphics[width=0.8\textwidth]{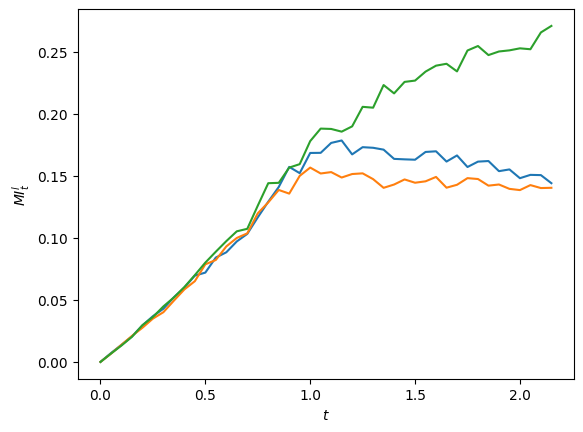}
  \caption{Pathwise market impact simulations for limit orders when $\alpha = 0.6$, $\lambda = 1$ and $\gamma = 0.5$.}
  \label{fig:limit:shape}
\end{figure}



\section{Calibration of $\kappa$}
\label{sec:estimation}

\subsection{$\kappa$ seen as a constant}
\label{sec:estimation:constant}

In \cite{jaisson2015market}, the market price, denoted by $P_t$, satisfies
\begin{equation*}
P_t 
= 
P_0 + \kappa 
\int_0^t 
\xi(t-s)  d (N^{a}_s - N^{b}_s)
\end{equation*}
where $\xi(t) = 1 + ( 1 + \int_0^\infty \psi(s)  ds ) \int_t^\infty \varphi(s)  ds$ and $\psi = \sum_{k\geq 1} \varphi^{*k}$ where $\varphi^{*k}$ stands for the $k$ fold convolution of $\varphi$. In this model, we see that a market order arriving at time $t_0$ pushes instantly the price by $\pm \kappa  \xi(0)$, then its impact at $t > t_0$ is $\pm \kappa  \xi(t - t_0)$. 
In particular, note that $\xi(t) \to 1$ as $t \to \infty$ and therefore the permanent impact of a single order is given by $\kappa$. In that view, $\kappa$ can be seen as a measure of the information content of a given trade. 

In this model, the quadratic variation of the price is given by $\kappa^2  \xi(0)^2 (N^{a}_t + N^{b}_t)$. This is due to the fact that when a market order arrives, the price jumps by $\pm \kappa  \xi(0)$, giving a contribution $\kappa^2  \xi(0)^2$ to the quadratic variation. In between these jumps, the price change only comes from the decay of the function $\xi$, which is differentiable, and thus does not contribute to the quadratic variation. Over a long time period, this quantity is related to the realised variance of stochastic volatility models. From a modeling perspective, the price is often assumed to follow a semi-martingale process
\begin{equation*}
  dX_t = b_t  dt + \sigma_t  dB_t
\end{equation*}
where $b_t$ is the drift, $\sigma_t$ is the volatility and $B$ is a Brownian motion. In that case, the realised variance $\int_0^t \sigma_s^2  ds$ can be estimated using discrete price increments. Let $\delta_n > 0$ be the sampling frequency of the price. Following for instance  A\"it-Sahalia and Jacod \cite{ait2014high}, we have
\begin{equation*}
  \int_0^t \sigma_s^2  ds 
  \approx
  \sum_{k=1}^{\lfloor t / \delta_n \rfloor} (X_{k\delta_n} - X_{(k-1)\delta_n})^2.
\end{equation*}

This relation can be used to estimate $\kappa$. By choosing an appropriate time horizon $T$ and a sampling frequency $\delta_n$, we have
\begin{equation*}
  \kappa^2  \xi(0)^2 (N^{a}_T + N^{b}_T) \approx \sum_{k=1}^{\lfloor T / \delta_n \rfloor} (X_{k\delta_n} - X_{(k-1)\delta_n})^2
\end{equation*}
and therefore we estimate $\kappa$ by
\begin{equation*}
  \wh{\kappa} 
  =
  \bigg(
  \frac{1}{\xi(0)^2 (N^{a}_T + N^{b}_T)}
  \sum_{k=1}^{\lfloor T / \delta_n \rfloor} (X_{k\delta_n} - X_{(k-1)\delta_n})^2
  \bigg)^{1/2}
\end{equation*}
Note that $\xi(0) = (1 - a_T)^{-1}$ and can thus be estimated by fitting a Hawkes kernel on real market data.

\begin{remark}
    It is nowadays accepted that modern financial markets are highly endogenous. In the market price model of \cite{jaisson2015market}, this is illustrated by the fact that most orders in the Hawkes order flow are due to its self-exciting nature rather than the Poissonian baseline. In other words, the endogeneity ratio
\begin{equation*}
  \frac{
  \sum_{k=1}^\infty \norm{\varphi}^k_{L^1}
  }{
  1 + \sum_{k=1}^\infty \norm{\varphi}^k_{L^1}
  }
  =
  \norm{\varphi}_{L^1}
\end{equation*}
is close to $1$.

However, this model only accounts for active trading using market orders. As shown in this paper, this is not the only way to interact with the price, as passive trading also plays a role. To statistically account for this effect, one could introduce an additional Brownian motion $\wt B$ independent of $N^{a}$ and $N^{b}$ and set
\begin{equation*}
P_t 
=
P_0 + \kappa 
\int_0^t 
\xi(t-s)  d (N^{a}_s - N^{b}_s) + \wt\sigma \wt B_t
\end{equation*}
for some constant $\wt \sigma$. In this case, the realised volatility on $[0,T]$ would be $\kappa^2 (N^{a}_T + N^{b}_T) + \wt\sigma^2 T$. Following the same approach as before, we have
\begin{equation*}
  \kappa^2  \xi(0)^2 (N^{a}_T + N^{b}_T) + \wt\sigma^2 T \approx \sum_{k=1}^{\lfloor T / \delta_n \rfloor} (X_{k\delta_n} - X_{(k-1)\delta_n})^2.
\end{equation*}
The quantities $(N^{a}_T + N^{b}_T)$ and $T$ being observed from market data, we can do a linear regression with intercept to estimate $\kappa^2  \xi(0)^2$ and $\wt \sigma^2$.
\end{remark}

\subsection{Dependence in the limit order book state}

In addition to previous effects, it is now well established that the future price move depends on the current limit order book state \cite{burghardt2006measuring,gould2016queue,lehalle2017limit,stoikov2018micro}. In Section \ref{sec:model}, we model this dependence by allowing $\kappa$ to depend on the aggregated volume available on the limit order book. We can adapt the methodology of Section \ref{sec:estimation:constant} to give some statistical evidence that an increase of volatility decreases the value of $\kappa$. In the following, we consider a simplified version of the model of Section \ref{sec:model} where the price is given by
\begin{equation}
\label{eq:simplified}
  P_t = P_0 + \int_0^t \kappa(q^a_s) \xi(t-s) dN^a_s - \int_0^t \kappa(q^b_s) \xi(t-s) dN^b_s.
\end{equation}
Equation \eqref{eq:simplified} can be seen as a simplified version of the model of Section \ref{sec:model} where two main effects are removed. First, in \eqref{eq:simplified}, market orders are the only contributors to the quadratic variation of the price. The effects of limit and cancel orders are removed and only impact the price through the cumulative volumes $q^a$ and $q^b$. Secondly, the conditional expectation defining the price in \eqref{eq:price:lim} is removed and replaced by a formula analogous to the constant $\kappa$ case studied in \cite{jaisson2015market}. In this model, the quadratic variation of the price is given by 
\begin{equation*}
  \xi(0)^2 \bigg( \sum_{dN^a_t = 1} \kappa^2(q_t^a) + \sum_{dN^b_t = 1} \kappa^2(q_t^b) \bigg).
\end{equation*}

Although a non-parametric estimation of $\kappa$ would be preferable, it is out of the scope of this paper. Instead, we assume that $\kappa^2$ has the very simplified form
\begin{equation*}
  \kappa^2(q) = a + b q.
\end{equation*}
Following the same approach as before, we have
\begin{equation*}
  a \xi(0)^2 (N^a_t + N^a_t) + b \xi(0)^2 \bigg( \sum_{dN^a_t = 1} q_t^a + \sum_{dN^b_t = 1} q_t^b \bigg)
  \approx \sum_{k=1}^{\lfloor T / \delta_n \rfloor} (X_{k\delta_n} - X_{(k-1)\delta_n})^2.
\end{equation*} 

We can then estimate $a$ and $b$ by a linear regression. 



\section*{Acknowledgments}
Mathieu Rosenbaum and Gr\'egoire Szymanski gratefully acknowledge the financial support of the \'Ecole Polytechnique chairs {\it Deep Finance and Statistics} and {\it Machine Learning and Systematic Methods}.

\section*{Competing Interests}
The authors declare no competing interests.


\bibliographystyle{spmpsci}      
\bibliography{full_library}

\newpage
\appendix

\section{Summary about Hawkes processes}

In this section, we summarize some results about Hawkes processes with time-varying baseline. The proofs are omitted for conciseness. They can however be easily adapted from the constant baseline case, see for instance El Euch and Rosenbaum \cite{eleuch2019characteristic}.

\begin{definition}
A Hawkes process with baseline (or background rate) $\mu : [0,\infty) \to [0,\infty)$ and self-exciting kernel $\varphi : [0,\infty) \to \mathbb{R}$ is a process $N$ adapted to some filtration $(\mathcal{F}_t)_t$ such that the compensator $\Lambda$ of $N$ has the form $\Lambda_t = \int_0^t \lambda_s ds$ where
\begin{equation*}
\lambda_t = \mu_t + \int_0^{t-} \varphi (t-s) dN_s.
\end{equation*}
\end{definition}

\begin{lemma}
\label{lemma:expressions_M}
Define $M = N-\Lambda$ and $\psi=\sum_{k \geq 1} \varphi^{* k}$ where $\varphi^{* k}$ stands for the $k$ fold convolution of $\varphi$. Then for any $0 \leq t \leq T$, we have
\begin{align*}
\lambda_t &= \mu_t + \int_0^t \psi(t-s) \mu_s  ds + \int_0^{t-} \psi(t-s) dM_s,
\\
\int_0^t \lambda_s ds 
&=
\int_0^t \mu_s  ds
+\int_0^t \psi(t-s) \int_0^s \mu_u  du ds
+\int_0^t \psi(t-s) M_s  ds.
\end{align*}
\end{lemma}

\begin{lemma}
\label{lemma:m1}
For any $0 \leq t \leq T$, we have
\begin{align*}
\mathbb{E}[\lambda_t] &= \mu_t + \int_0^t \psi(t-s) \mu_s  ds,
\\
\mathbb{E}[N_t] &= \int_0^t \mu_s  ds
+\int_0^t \psi(t-s) \int_0^s \mu_u  du ds.
\end{align*}
\end{lemma}

\begin{lemma}
\label{lemma:m2}
For any $0 \leq t \leq T$, we have
\begin{align*}
\mathbb{E}[\lambda_t^2] &= \mathbb{E}[\lambda_t]^2 + \int_0^t \psi(t-s)^2 \mathbb{E}[\lambda_s] ds,
\\
\mathbb{E}[\Lambda_t^2] &= \mathbb{E}[\Lambda_t]^2 + \int_0^t \int_0^t \psi(t-u) \psi(t-v) \mathbb{E}[N_{u\wedge v}] dudv.
\end{align*}
\end{lemma}

\begin{corollary}
\label{cor:bound:Ninfty}
Suppose that $\mu \in L^1$ and that $\norm{\varphi}_{L^1} < 1$. Then $\E[N_\infty]$ and $\E[N_\infty^2]$ are finite.
\end{corollary}

\begin{proof}
We start we $\E[N_\infty]$. By monotone convergence theorem, it is enough to prove that $\E[N_t]$ is bounded by a constant uniformly in $t$. Using Lemma \ref{lemma:m1}, we have
\begin{equation*}
\mathbb{E}[N(t)] 
= \int_0^t \mu_s  ds +\int_0^t \psi(t-s) \int_0^s \mu_u  du ds
\leq \norm{\mu}_{L^1} + \norm{\psi}_{L^1} \norm{\mu}_{L^1}
\end{equation*}
which is bounded because $\norm{\psi}_{L^1} = \norm{\varphi}_{L^1} / (1-\norm{\varphi}_{L^1})^{-1} < \infty$ since $\norm{\varphi}_{L^1} < 1$.

We now study $\E[N_\infty^2]$. Similarly, it is enough to prove that $\E[N_t^2]$ is uniformly bounded. Using that $M$ is a martingale and $\langle M \rangle = N$, we have
\begin{equation*}
\mathbb{E}[N_t^2] 
\leq 2 (\mathbb{E}[M_t^2] + \mathbb{E}[\Lambda_t^2]) \leq 2 \bigg(\mathbb{E}[N_t] + \mathbb{E}[\Lambda_t]^2 + \int_0^t \int_0^t \psi(t-u) \psi(t-v) \mathbb{E}[N_{u\wedge v}]  dudv \bigg).
\end{equation*}
Then, using that for all $u,v \leq t$, $\mathbb{E}[N_{u\wedge v}] \leq \mathbb{E}[N_t] = \mathbb{E}[\Lambda_t]$ is bounded, we have
\begin{equation*}
\mathbb{E}[N_t^2] 
\leq 2 (\mathbb{E}[M_t^2] + \mathbb{E}[\Lambda_t^2])
\leq C \bigg(1 + \int_0^t \int_0^t \psi(t-u) \psi(t-v)  dudv \bigg)
\\
= C \bigg(1 + \bigg(\int_0^t \psi(u)  du \bigg)^2 \bigg),
\end{equation*}
and we conclude using that $\norm{\psi}_{L^1} < \infty$.
\end{proof}

\section{Proof of Theorem \ref{thm:finiteness:price:lim}}
\label{proof:thm:finiteness:price:lim}

\subsection{Outline of the proof}

The main idea of the proof is to use the fact that although a Hawkes process is not Markovian, since it depends on the entire history of arrivals, we can lift it and see it as a Markov process in a bigger state space, see Cuchiero and Teichmann \cite{cuchiero2019markovian,cuchiero2020generalized}. More precisely, we have the following result.
\begin{lemma}
\label{lemma:lift:hawkes}
Let $N$ be a Hawkes process with baseline $\mu$ and self-exciting kernel $\varphi$ and we write $\calF_t$ the natural filtration of the Hawkes process. For $t \geq 0$, the distribution of $(N_{s+t})_{s\geq 0}$ given $\calF_t$ is the same as $N_t + \wt N + \wh N$, where $\wt N$ and $\wh N$ have the following properties
\begin{enumerate}
\item $\wt N$ and $\wh N$ are two independent processes,
\item $\wt N$ is independent of $\calF_t$ and is a Hawkes process with the same distribution as $N$,
\item Conditionally on $\calF_t$, $\wh N$ is a Hawkes process with self-exciting kernel $\varphi$ and time-varying baseline 
\begin{equation*}
\wh \mu_s = \int_0^{t} \varphi(t+s-u) d N_u.
\end{equation*}
\end{enumerate}
\end{lemma}

We apply this result to $N^a$ and $N^b$. We introduce 
\begin{equation*}
\wh \mu^a_s = \int_0^{t} \varphi(t+s-u)  dN^a_u
\quad \text{ and } \quad 
\wh \mu^b_s = \int_0^{t} \varphi(t+s-u)  dN^b_u.
\end{equation*}
and we consider four processes $(\wh N^{a}, \wh N^{b}, \wt N^{a}, \wt N^{b})$ with the following properties:
\begin{enumerate}
\item The processes $(\wh N^{a}, \wh N^{b}, \wt N^{a}, \wt N^{b})$ are independent,
\item $\wt N^{a}$ and $\wt N^{b}$ are independent of $\calF_t$ and are two Hawkes processes with the same distribution as $N^{a}$ and $N^{b}$,
\item Conditionally on $\calF_t$, $\wh N^{a}$ and $\wh N^{b}$ are two Hawkes processes with self-exciting kernel $\varphi$ and time-varying baseline $\wh \mu^a_s$ and $\wh \mu^b_s$.
\end{enumerate}
Therefore, conditionally on $\calF_t$, $(N^a_{s+t}, N^b_{s+t})$ has the same distribution as $(N^a_t + \wt N^a_{s} + \wh N^a_{s}, N^b_t + \wt N^b_{s} + \wh N^b_{s})$. Moreover, note that given $(N^a, N^b)$, the dynamic of $(q^a, q^b)$ is Markovian. Therefore, conditionally on $\calF_t$, the dynamic of $(N^a_{s+t}, N^b_{s+t}, q^a_{s+t}, q^b_{s+t})$ is the same as 
$(N^a_t + \wt N^a_{s} + \wh N^a_{s}, N^b_t + \wt N^b_{s} + \wh N^b_{s}, \wt q^a_s, \wt q^b_s)$ where 
\begin{equation*}
\wt q^a_s = q^a_t + \wt L^{a}_s - \wt C^{a}_s - \big(\wt N^a_{s} + \wh N^a_{s} \big)
\end{equation*}
and 
\begin{equation*}
\wt q^b_s = q^b_t + \wt L^{b}_s - \wt C^{b}_s - \big(\wt N^b_{s} + \wh N^b_{s} \big)
\end{equation*}
and where $(\wt L^{a}, \wt L^{b}, \wt C^{a}, \wt C^{b})$ are time inhomogeneous Poisson processes given by
\begin{equation*}
\begin{split}
\wt L^{a}_s &= \int_0^s \int_0^\infty \1_{\{\wt \lambda^{L, a}_u \leq z\}} \wt \pi^{L, a}(du, dz),
 \quad \quad
\wt L^{b}_s = \int_0^s \int_0^\infty \1_{\{\wt \lambda^{L, b}_u \leq z\}} \wt \pi^{L, b}(du, dz), \\
\wt C^{a}_s &= \int_0^s \int_0^\infty \1_{\{\wt \lambda^{C, a}_u \leq z\}} \wt \pi^{C, a}(du, dz),
 \quad \quad
\wt C^{b}_s = \int_0^s \int_0^\infty \1_{\{\wt \lambda^{C, b}_u \leq z\}} \wt \pi^{C, b}(du, dz)
\end{split}
\end{equation*}
for some independent Poisson point measures $(\wt \pi^{L, a}, \wt \pi^{L, b}, \wt \pi^{C, a}, \wt \pi^{C, b})$ which are also independent from $(\wh N^{a}, \wh N^{b}, \wt N^{a}, \wt N^{b})$ and of $\calF_t$; and with intensities 
\begin{equation*}
\wt \lambda^{L, a}_s = \lambda^{L}(\wt q^a_{s-}), \quad \quad
\wt \lambda^{L, b}_s = \lambda^{L}(\wt q^b_{s-}), \quad \quad
\wt \lambda^{C, a}_s = \lambda^{C}(\wt q^a_{s-}),
\quad \text{ and } \quad
\wt \lambda^{C, b}_s = \lambda^{C}(\wt q^b_{s-}).
\end{equation*}
Combining all these definitions, we have for all $T > t$
\begin{align*}
\E\bigg[\int_0^T \kappa(q^a_s)  dN^a_s \bigg|  \calF_t \bigg]
&=
\int_0^t \kappa(q^a_s)  dN^a_s
+
\E\bigg[\int_t^T \kappa(q^a_s)  dN^a_s \bigg|  \calF_t \bigg]
\\
&=
\int_0^t \kappa(q^a_s)  dN^a_s
+
\E\bigg[\int_0^{T-t} \kappa(\wt q^a_s)  d\wt N^a_s + \int_0^{T-t} \kappa(\wt q^a_s)  d\wh N^a_s \bigg|  \calF_t \bigg]
\end{align*}
and, similarly, we have
\begin{equation*}
\E\bigg[\int_0^T \kappa(q^b_s)  dN^b_s \bigg|  \calF_t \bigg]
=
\int_0^t \kappa(q^b_s)  dN^b_s
+
\E\bigg[\int_0^{T-t} \kappa(\wt q^b_s)  d\wt N^b_s + \int_0^{T-t} \kappa(\wt q^b_s)  d\wh N^b_s \bigg|  \calF_t \bigg].
\end{equation*}
Replacing $T$ by $T+t$, we see that the proof of Theorem \ref{thm:finiteness:price:lim} results in proving the convergence when $T \to \infty$ of
\begin{equation}
\label{eq:limit:simplify}
\E\bigg[\int_0^{T} \kappa(\wt q^a_s)  d\wt N^a_s + \int_0^{T} \kappa(\wt q^a_s)  d\wh N^a_s \bigg|  \calF_t \bigg]
- 
\E\bigg[\int_0^{T} \kappa(\wt q^b_s)  d\wt N^b_s + \int_0^{T} \kappa(\wt q^b_s)  d\wh N^b_s \bigg|  \calF_t \bigg].
\end{equation}

The idea now is to replace the random variables in the second expectation with identically distributed ones that are coupled with those in the first expectation, specifically, they share the same Poisson point process driving their jump times. Note first that $\wt N^b$ and $\wt N^a$ are two independent Hawkes processes with the common baseline and self-exciting kernel and that they are both independent from $\calF_t$. Similarly, $(\wt \pi^{L, b}, \wt \pi^{C, b})$ and $(\wt \pi^{L, a}, \wt \pi^{C, a})$ are independent, independent from $\calF_t$ and have the same distribution. Therefore, we have
\begin{equation}
\label{eq:rewriting:coupling}
\E\bigg[\int_0^{T} \kappa(\wt q^b_s)  d\wt N^b_s+ \int_0^{T} \kappa(\wt q^b_s)  d\wh N^b_s \bigg|  \calF_t \bigg]
=
\E\bigg[\int_0^{T} \kappa(\wt q^{b, \prime}_s)  d\wt N^a_s+ \int_0^{T} \kappa(\wt q^{b, \prime}_s)  d \wh N^b_s \bigg|  \calF_t \bigg]
\end{equation}
where 
\begin{equation*}
\wt q^{b, \prime}_s = q^b_t + \wt L^{b, \prime}_s - \wt C^{b, \prime}_s - \big(\wt N^a_{s} + \wh N^b_{s} \big)
\end{equation*}
and where we write
\begin{equation*}
\wt L^{b, \prime}_s = \int_0^s \int_0^\infty \1_{\{\wt \lambda^{L, b, \prime}_u \leq z\}} \wt \pi^{L, a}(du, dz) 
\quad \text{ and } \quad
\wt C^{b, \prime}_s = \int_0^s \int_0^\infty \1_{\{\wt \lambda^{C, b, \prime}_u \leq z\}} \wt \pi^{C, a}(du, dz)
\end{equation*}
with
\begin{equation*}
\wt \lambda^{L, b, \prime}_s = \lambda^{L}(\wt q^{b\prime}_{s-}) 
\quad \text{ and } \quad
\wt \lambda^{C, b, \prime}_s = \lambda^{C}(\wt q^{b\prime}_{s-}).
\end{equation*}

Note in addition that
\begin{equation*}
\int_0^t \wh \mu^a_s  ds \leq \norm{\varphi}_{L^1} N^a_t
\quad \text{ and } \quad
\int_0^t \wh \mu^b_s  ds \leq \norm{\varphi}_{L^1} N^b_t
\end{equation*}
which is almost surely finite. Combining \eqref{eq:limit:simplify} and \eqref{eq:rewriting:coupling}, we see that the proof of Theorem \ref{thm:finiteness:price:lim} is completed once we prove the convergence when $T\to \infty$ of the quantity
\begin{equation*}
\E\bigg[\int_0^{T} \kappa(\wt q^a_s)  d (\wt N^a_s + \wh N^a_s) \bigg|  \calF_t \bigg]
- 
\E\bigg[\int_0^{T} \kappa(\wt q^{b, \prime}_s)  d(\wt N^a_s + \wh N^b_s ) \bigg|  \calF_t \bigg].
\end{equation*}
We conclude using the description of the law of $(\wt N^a_s,  \wh N^a_s, \wh N^b_s)$ conditionally to $\calF_t$ and the following Proposition, proved in Section \ref{sec:proof:lemma:complete:finiteness:price:lim}. 

\begin{proposition}
\label{lemma:complete:finiteness:price:lim}
Let $f$ and $g$ be two $L^1$ functions. We consider three independent Hawkes processes $N$, $N^f$ and $N^g$ with baseline $\mu$ (constant), $f$ and $g$. We consider also $q^f$ an $q^g$ such that 
\begin{equation*}
\begin{cases}
q^f_t 
=
q^f_0 
+
\int_0^t \int_0^\infty
\1_{\{z \leq \lambda^{L}(q^f_{s-})\}}

\pi^{L}(ds,  dz) 
-
\int_0^t \int_0^\infty
\1_{\{z \leq \lambda^{C}(q^f_{s-})\}}

\pi^{C}(ds,  dz) 
-
N_t
-
N^f_t
\\
q^g_t 
=
q^g_0 
+
\int_0^t \int_0^\infty
\1_{\{z \leq \lambda^{L}(q^g_{s-})\}}

\pi^{L}(ds,  dz) 
-
\int_0^t \int_0^\infty
\1_{\{z \leq \lambda^{C}(q^g_{s-})\}}

\pi^{C}(ds,  dz) 
-
N_t
-
N^g_t
\end{cases}
\end{equation*}
Then the limit
\begin{equation*}
\lim_{T\to\infty}
\E
\bigg[
\int_0^T \kappa(q_s^f) d(N^f_s + N_s)
-
\int_0^T \kappa(q_s^g) d(N^g_s + N_s)
\bigg]
\end{equation*}
is well defined and finite.
\end{proposition}

\subsection{Proof of Proposition \ref{lemma:complete:finiteness:price:lim}}
\label{sec:proof:lemma:complete:finiteness:price:lim}
We first write
\begin{align*}
\E
\bigg[
\int_0^T \kappa(q_s^f) d(N^f_s + N_s)
-
\int_0^T \kappa(q_s^g) d(N^g_s + N_s)
\bigg]
=
&\E
\bigg[
\int_0^T \kappa(q_s^f) dN^f_s
\bigg]
-
\E
\bigg[\int_0^T \kappa(q_s^g) dN^g_s
\bigg]
\\
&+
\E
\bigg[\int_0^T \big( \kappa(q_s^f) - \kappa(q_s^g) \big) dN_s
\bigg].
\end{align*}
Note that $\int_0^T \kappa(q_s^f) dN^f_s$ is non-decreasing, non-negative and converges towards $\int_0^\infty \kappa(q_s^f) dN^f_s$. Using that $\kappa$ is bounded and non-negative, we also have
\begin{equation*}
\int_0^\infty \kappa(q_s^f) dN^f_s
\leq 
\norm{\kappa}_\infty N^f_\infty.
\end{equation*}
We introduce the stopping times $\tau_f$ and $\tau_g$ defined by 
\begin{equation*}
\begin{cases}
\tau^f = \sup \{t \geq 0: N_t^f \neq N_{t-}^f \},
\\
\tau^g = \sup \{t \geq 0: N_t^g \neq N_{t-}^g \}
\end{cases}
\end{equation*}
and we write $\tau = \tau^f \vee \tau^g$. Since $\norm{f}_{L^1} + \norm{q}_{L^1} < \infty$ , we know that $\tau$ is almost surely finite.

Since $N^f$ is a Hawkes process with a time varying baseline, we can compute explicitly its expectation and we can see using Corollary \ref{cor:bound:Ninfty} that 
$\E[N^f_\infty]$ is finite. 
Therefore monotone's convergence theorem yields
\begin{equation*}
\E
\bigg[
\int_0^T \kappa(q_s^f) dN^f_s
\bigg]
\to_{T\to\infty} 
\E
\bigg[
\int_0^{\tau^f} \kappa(q_s^f) dN^f_s
\bigg] < \infty.
\end{equation*}
Similarly, we have
\begin{equation*}
\E
\bigg[
\int_0^T \kappa(q_s^g) dN^g_s
\bigg]
\to_{T\to\infty} 
\E
\bigg[
\int_0^{\tau^f} \kappa(q_s^g) dN^g_s
\bigg] < \infty.
\end{equation*}
Therefore, it remains to prove that 
\begin{equation*}
\E
\bigg[
\int_0^T \big( \kappa(q_s^f) - \kappa(q_s^g) \big) dN_s
\bigg]
\end{equation*}
converges as $T\to\infty$.  To prove this limit, we first claim that that
\begin{equation}
\label{eq:finite:integral}
\E
\bigg[
\int_0^\infty | \kappa(q_s^f) - \kappa(q_s^g) |  dN_s
\bigg]
< \infty.
\end{equation}
Note that this expectation is always well defined because $| \kappa(q_s^f) - \kappa(q_s^g) |$ is nonnegative. Moreover, this limit the limit of the non-decreasing sequence $\E
[
\int_0^T | \kappa(q_s^f) - \kappa(q_s^g) | dN_s
]$.
We now complete the proof of Proposition \ref{lemma:complete:finiteness:price:lim} and we then prove that \eqref{eq:finite:integral} holds. For $T > 0$, we write 
$E_T = \E
[
\int_0^T ( \kappa(q_s^f) - \kappa(q_s^g) ) dN_s
]
$
and we want to prove that $(E_T)_T$ is a Cauchy sequence. Let $\varepsilon > 0$. Then there exists $R > 0$ such that 
\begin{equation*}
\E
\bigg[
\int_0^\infty | \kappa(q_s^f) - \kappa(q_s^g) |  dN_s
\bigg]
-
\E
\bigg[
\int_0^R | \kappa(q_s^f) - \kappa(q_s^g) |  dN_s
\bigg]
=
\E
\bigg[
\int_R^\infty | \kappa(q_s^f) - \kappa(q_s^g) |  dN_s
\bigg]
< \varepsilon.
\end{equation*}
Therefore, for $R \leq T_1 \leq T_2 < \infty$, we have
\begin{align*}
    |E_{T_1} - E_{T_2}| 
    &\leq
    \bigg|
    \E
    \bigg[
    \int_{T_1}^{T_2} \big( \kappa(q_s^f) - \kappa(q_s^g) \big) dN_s
    \bigg]
    \bigg|
\\    
    &\leq
    \E
    \bigg[
    \int_{T_1}^{T_2} | \kappa(q_s^f) - \kappa(q_s^g) | dN_s
    \bigg]
\\    
    &\leq
    \E
    \bigg[
    \int_{R}^{\infty} | \kappa(q_s^f) - \kappa(q_s^g) | dN_s
    \bigg]
\\    
    &\leq
    \varepsilon.
\end{align*}
Thus $(E_T)_T$ is a Cauchy sequence and thus converges as $T \to \infty$ towards a finite value.

It remains to prove that \eqref{eq:finite:integral} holds. We define $(\nu_1, \dots, \nu_{S})$ the jump times of $N^f + N^g$, where $S = N^f_\tau + N^g_\tau$. We write $\nu_0 = 0$ and $\nu_{S+1} = \infty$ for conciseness. We then claim the following.

\begin{lemma}
\label{lem:control:diff:q}
$|q^f_t - q^g_t|$ is non-increasing on all intervals $(\nu_i, \nu_{i+1})$. Moreover, for all $t \geq 0$, we have
\begin{equation*}
|q^f_t - q^g_t|
\leq 
|q^f_0 - q^g_0 | + N^f_t + N^g_t.
\end{equation*}

\end{lemma}

Intuitively, these results hold because of the properties of $\lambda^L$, $\lambda^C$ and because of the coupling between $q^f$ and $q^g$. A detailed proof can be found in Section \ref{sec:proof:remaininglemmas}. In particular, we see that $|q^f_t - q^g_t|$ is bounded by $M = S + |q^f_0 - q^g_0 |$ which is independent of $N$. 

We now define for each $k \geq 0$ and $0 \leq i \leq S$
\begin{equation*}
\tau_{i, k} = \inf \{t \geq \nu_i:  |q^f_t - q^g_t| \leq k \}.
\end{equation*}
so that we have for all $T \geq 0$
\begin{align*}
\int_0^\infty | \kappa(q_s^f) - \kappa(q_s^g) |  dN_s
&\leq \sum_{i=0}^S 
\bigg|
\int_{\nu_i}^{\nu_{i+1}} | \kappa(q_s^f) - \kappa(q_s^g) |  dN_s
\bigg|
\\
&\leq \sum_{i=0}^S 
\sum_{k\geq 1}
\int_{\tau_{i,k}}^{\tau_{i,k-1} \wedge \nu_{i+1}} | \kappa(q_s^f) - \kappa(q_s^g) | dN_s
\\
&\leq \sum_{i=0}^S 
\sum_{k\geq 1}
2 \norm{\kappa}_\infty
(N_{\tau_{i,k-1} \wedge \nu_{i+1}} - N_{\tau_{i,k}}).
\end{align*}
Note that since $|q^f_t - q^g_t|$ is bounded by $M$, we know that $\tau_{i,k} = \tau_{i,M}$ for all $k \geq M$. Thus, in the last line, the sum over $k \geq 1$ can be replaced by a sum over $1 \leq k \leq M$. The idea is then to use the fact that $N$ is a Hawkes process to have
\begin{equation*}
\E [ N_{\tau_{i,k-1} \wedge \nu_{i+1}} - N_{\tau_{i,k}} ] \leq C \big((\tau_{i,k-1} \wedge \nu_{i+1}) - \tau_{i,k}\big)
\end{equation*}
and then conclude by studying the difference $\tau_{i,k-1} \wedge \nu_{i+1} - \tau_{i,k}$. However, this does not work directly because $\tau_{i,k-1} \wedge \nu_{i+1}$ and $\tau_{i,k}$ are not independent of $N$. However, note that in the definition of $\tau_{i,k}$, we have
\begin{align*}
|q^f_t - q^g_t|
=
\bigg|
q^f_0 
&-
q^g_0 
+
\int_0^t \int_0^\infty
(
\1_{\{z \leq \lambda^{L}(q^f_{s-})\}}
-
\1_{\{z \leq \lambda^{L}(q^g_{s-})\}}
)
\pi^{L}(ds,  dz) 
\\
&-
\int_0^t \int_0^\infty
(
\1_{\{z \leq \lambda^{C}(q^f_{s-})\}}
-
\1_{\{z \leq \lambda^{C}(q^g_{s-})\}}
)
\pi^{C}(ds,  dz) 
-
N^f_t
+
N^g_t
\bigg|
\end{align*}
and therefore, $N$ only intervenes via the intensity within the jump integrals. Therefore, using crucially the assumption of Theorem \ref{thm:finiteness:price:lim}, we have the following result, proved in Section \ref{sec:proof:remaininglemmas}
\begin{lemma}
\label{lem:expvar}
There exists a family of independent random variables $(\nu_{i,k})$, independent from $N$, such that for all $0 \leq i \leq S$ and all $k \geq M$, $\nu_{i,k}$ follows an exponential distribution with parameter $\mathfrak{m}_k$, and we have
\begin{equation*}
(\tau_{i,k-1} \wedge \nu_{i+1}) - \tau_{i,k} \leq \nu_{i,k}.
\end{equation*}
\end{lemma}
Moreover, if $\nu_1\leq\nu_2$ are two variables $\calC$-measurable and if $\calC$ is independent of $N$, then we have
\begin{equation*}
\E[ N_{\nu_2} - N_{\nu_1}|\calC ] \leq C (\nu_2 - \nu_1) 
\end{equation*}
for some constant $C > 0$ because the expectation of $\lambda_t$ is bounded independently of $t$. Therefore we have, using that $N$ independent of both $M$, $S$ and $(\nu_{i,k})_{i,k}$, we have
\begin{align*}
\E\bigg[
\int_0^\infty | \kappa(q_s^f) - \kappa(q_s^g) |  dN_s
\Big]
&\leq 
2 C \norm{\kappa}_\infty
\E \bigg[
\sum_{i=0}^S 
\sum_{k = 1}^{M}
\nu_{i,k}
\bigg]
\end{align*}
for some constant $C > 0$. Then, using that $\nu_{i,k}$ is independent of $S$ and $M$ and its expectation is $\mathfrak{m}_k^{-1}$ which is bounded uniformly for $k$, we have
\begin{align*}
\E\bigg[
\int_0^\infty | \kappa(q_s^f) - \kappa(q_s^g) |  dN_s
\bigg]
&\leq C^{\prime} 
\E [
SM
]
\end{align*}
for some constant $C^{\prime} > 0$. We conclude using that $\E [SM]$ is finite since $\E[N^f]$, $\E[N^g]$, $\E[(N^f)^2]$ and $\E[(N^g)^2]$ are bounded by Corollary \ref{cor:bound:Ninfty}. This proves that \eqref{eq:finite:integral} holds.

\subsection{Proof of the remaining lemmas}
\label{sec:proof:remaininglemmas}

\subsubsection*{Proof of Lemma \ref{lemma:lift:hawkes}}
The proof is classical and follows the ideas of the branching tree representation of the Hawkes processes. By Theorem~7.4 \cite{ikeda1989stochastic}, there exists a Poisson point process $\pi$ on $[0,\infty) \times [0,\infty)$ with compensator $dsdz$ such that for all $t \geq 0$
\begin{equation*}
  N_t = \int_0^t \int_0^\infty \1_{\{z \leq \lambda_{s-}\}}  \pi(ds,  dz)
\end{equation*}
where $\lambda$ is the intensity of the Hawkes process $N$. Fix $t \geq 0$. Recall that $\wh \mu$ is defined within Lemma \ref{lemma:lift:hawkes}. 
For $s \geq 0$, we define then
\begin{equation*}
  \wh N_s = \int_{0+}^{s} \int_0^\infty \1_{\{z \leq \wh \lambda_{u-}\}}  \pi(du + t,  dz)
\end{equation*}
where $\wh \lambda$ is defined by
\begin{equation*}
  \wh \lambda_s = \wh \mu_s + \int_0^s \varphi(s-u)  d \wh N_u.
\end{equation*}
Note that $\wh \lambda_s \leq \lambda_{t+s}$ for all $s \geq 0$. Therefore we can define $\wt N_s = N_{t+s} - N_t - \wh N_s$ for all $s \geq 0$ so that $N_{t+s} = N_t + \wt N_s + \wh N_s$. Moreover, we have
\begin{align*}
  \wt N_s 
  &= 
  \int_{t+}^{t+s} \int_0^\infty \1_{\{z \leq \lambda_{u-}\}}  \pi(du + t,  dz)
  - 
  \int_{t+}^{t+s} \int_0^\infty \1_{\{z \leq \wh \lambda_{u-}\}}  \pi(du + t,  dz)
  \\
  &=
  \int_{t+}^{t+s} \int_0^\infty \1_{\{\wh \lambda_{u-} < z \leq \lambda_{u-}\}}  \pi(du + t,  dz).
\end{align*}
Since $\pi$ is a Poisson random measure, we can build $\wt \pi$ and $\wh \pi$ two independent Poisson point measures such that $\wt \pi(du,  dz) = \pi(du + t ,  dz + \wh \lambda_{u-})$ and $\wh \pi(du,  dz) = \pi(du + t,  dz) \1_{\{z \leq \wh \lambda_{u-}\}} + \pi'(du,  dz) \1_{\{z > \wh \lambda_{u-}\}}$ for some Poisson point measure $\pi'$ independent of $\pi$. With these notations, we have 
\begin{equation*}
  \wh N_s = \int_{0+}^{s} \int_0^\infty \1_{\{z \leq \wh \lambda_{u-}\}}  \wh \pi(du,  dz)
\end{equation*}
and
\begin{align*}
  \wt N_s 
  &=
  \int_{0}^{s} \int_0^\infty \1_{\{z \leq \lambda_{u-} - \wh \lambda_{u-}\}}  \wt \pi(du,  dz).
\end{align*}
Defining also $\wt \lambda_s = \lambda_{s} - \wh \lambda_{s} $, we have
\begin{align*}
  \wt \lambda_s 
  &= \mu + \int_0^{t+s} \varphi(t+s-u)  dN_u - \wh \mu_s - \int_0^s \varphi(s-u)  d \wh N_u
  \\ 
  &= \mu + \int_{0}^{s} \varphi(s-u)  dN_{t+u} - \int_0^s \varphi(s-u)  d \wh N_u
  \\ 
  &= \mu + \int_{0}^{s} \varphi(s-u)  d \wt N_{u}.
\end{align*}
This ensures that $\wt N$ is a Hawkes process with the same dynamic as $N$, and $\wh N$ is a Hawkes process with time-varying baseline $\wh \mu$. Independence properties follow from the properties of the Poisson point processes.

\subsubsection*{Proof of Lemma \ref{lem:control:diff:q}}

Recall that we have
\begin{align*}
|q^f_t - q^g_t|
=
\bigg|
q^f_0 
&-
q^g_0 
+
\int_0^t \int_0^\infty
(
\1_{\{z \leq \lambda^{L}(q^f_{s-})\}}
-
\1_{\{z \leq \lambda^{L}(q^g_{s-})\}}
)
\pi^{L}(ds,  dz) 
\\
&-
\int_0^t \int_0^\infty
(
\1_{\{z \leq \lambda^{C}(q^f_{s-})\}}
-
\1_{\{z \leq \lambda^{C}(q^g_{s-})\}}
)
\pi^{C}(ds,  dz) 
-
N^f_t
+
N^g_t
\bigg|.
\end{align*}
Moreover, by definition $N^f$ and $N^g$ are constant on $(\nu_i, \nu_{i+1})$. Therefore, any change in $|q^f_t - q^g_t|$ must come from one of the two jump integrals. Let $t$ be one of the jumping times of $|q^f_t - q^g_t|$. We can distinguish three cases depending on the sign of $q^f_t - q^g_t$.
\begin{itemize}
\item Suppose first that $q^f_{t-} > q^g_{t-}$. Then, since $\lambda^{L}$ is decreasing and $\lambda^{C}$ in increasing, we have
\begin{equation*}
\lambda^{L}(q^f_{t-}) \leq \lambda^{L}(q^g_{t-})
\quad \text{ and } \quad
\lambda^{C}(q^f_{t-}) \geq \lambda^{C}(q^g_{t-}).
\end{equation*}
Thus we must have $ q^f_t - q^g_t = q^f_{t-} - q^g_{t-} - 1$ and therefore $|q^f_t - q^g_t| = |q^f_{t-} - q^g_{t-}| - 1$.
\item Suppose now that $q^f_{t-} = q^g_{t-}$. Then it is clear that neither of the two jump integrals can jump and therefore $t$ cannot be a jumping time of $|q^f_t - q^g_t|$.
\item Suppose now that $q^f_{t-} < q^g_{t-}$. Then we can apply the same argument as in the case $q^f_{t-} > q^g_{t-}$ and we see that $ q^f_t - q^g_t = q^f_{t-} - q^g_{t-} + 1$ so that we still have $|q^f_t - q^g_t| = |q^f_{t-} - q^g_{t-}| - 1$.
\end{itemize}

Therefore, $|q^f_t - q^g_t|$ is decreasing on all intervals of the form $(\nu_i, \nu_{i+1})$. Moreover, for each $1 \leq i \leq S$, we have
\begin{equation*}
q^f_{\nu_i} - q^g_{\nu_i} = q^f_{\nu_i -} - q^g_{\nu_i - } \pm 1
\end{equation*}
depending on whether $N^f$ or $N^g$ jump and therefore $|q^f_{\nu_i} - q^g_{\nu_i} | \leq |q^f_{\nu_i -} - q^g_{\nu_i - }| + 1$. By induction, we obtain
\begin{equation*}
|q^f_t - q^g_t|
\leq 
|q^f_0 - q^g_0 | + N^f_t + N^g_t.
\end{equation*}

\subsubsection*{Proof of Lemma \ref{lem:expvar}}
Without loss of generality, suppose that $q^f_t \geq q^g_t$ on $(\nu_i, \nu_{i+1})$. Then, we have for $\tau_{i,k} \leq t \leq \nu_{i+1}$
\begin{align*}
q^f_t - q^g_t
=
k
&+
\int_{\tau_{i,k}}^t \int_0^\infty
(
\1_{\{z \leq \lambda^{L}(q^g_{s-} + k)\}}
-
\1_{\{z \leq \lambda^{L}(q^g_{s-})\}}
)
\pi^{L}(ds,  dz) 
\\
&-
\int_{\tau_{i,k}}^t \int_0^\infty
(
\1_{\{z \leq \lambda^{C}(q^g_{s-} + k)\}}
-
\1_{\{z \leq \lambda^{C}(q^g_{s-})\}}
)
\pi^{C}(ds,  dz).
\end{align*}
Recall that we have
\begin{equation*}
\lambda^{L}(q^g_{s-}) \geq \lambda^{L}(q^g_{s-} + k)
\quad \text{ and } \quad
\lambda^{C}(q^g_{s-} + k) \geq \lambda^{C}(q^g_{s-})
\end{equation*}
because $\lambda^{L}$ is decreasing and $\lambda^{C}$ is increasing. Therefore, there exists a Poisson point measures $\wt \pi^L$ and $\wt \pi^C$ such that 
\begin{equation*}
\begin{split}
\int_{\tau_{i,k}}^t \int_0^\infty
(
\1_{\{z \leq \lambda^{L}(q^g_{s-} + k)\}}
-
\1_{\{z \leq \lambda^{L}(q^g_{s-})\}}
)
\pi^{L}(ds,  dz) 
&=
-
\int_{\tau_{i,k}}^t \int_0^\infty
\1_{\{z \leq \lambda^{L}(q^g_{s-}) - \lambda^{L}(q^g_{s-} + k)\}}
\wt \pi^{L}(ds,  dz) 
\\
\int_{\tau_{i,k}}^t \int_0^\infty
(
\1_{\{z \leq \lambda^{C}(q^g_{s-} + k)\}}
-
\1_{\{z \leq \lambda^{C}(q^g_{s-})\}}
)
\pi^{C}(ds,  dz)
&=
\int_{\tau_{i,k}}^t \int_0^\infty
\1_{\{z \leq \lambda^{C}(q^g_{s-} + k) - \lambda^{C}(q^g_{s-})\}}
\wt \pi^{C}(ds,  dz).
\end{split}
\end{equation*}
We then merge these two jump integrals: there exists a Poisson point measure $\wt \pi$ such that
\begin{equation*}
q^f_t - q^g_t
=
k
-
\int_{\tau_{i,k}}^t \int_0^\infty
\1_{\{z \leq \lambda^{L}(q^g_{s-}) - \lambda^{L}(q^g_{s-} + k) + \lambda^{C}(q^g_{s-} + k) - \lambda^{C}(q^g_{s-})\}}
\wt \pi(ds,  dz).
\end{equation*}

We can then conclude using that $\lambda^{L}(q^g_{s-}) - \lambda^{L}(q^g_{s-} + k) + \lambda^{C}(q^g_{s-} + k) - \lambda^{C}(q^g_{s-})$ is bounded below by $ \mathfrak{m}_k$ and using that the first jumping time of a Poisson point process with intensity $\mathfrak{m}_k$ follows an exponential distribution with parameter $\mathfrak{m}_k$. The independence between the different exponential variable obtained is due to the independence properties of Poisson point measures.

\section{Proof of Proposition \ref{prop:convergence:LOB}}
\label{sec:proof:prop:convergence:LOB}

Without loss of generality, we prove Proposition \ref{prop:convergence:LOB} on the time interval $[0,K]$. Note that since $q^{a,T}$ and $q^{a,T}$ are independent, we can study the convergence of each sequence separately. In the following, we study one of this sequence, and we drop the exponent $x \in \{a,b\}$ to ease the notations. the proof is split into four parts:
\begin{itemize}
\item Step 1: We show that $(\wh q_t^T)_t$ is bounded above in probability.
\item Step 2: We show that $(\wh C_t^T)_t$ is tight for the Skorokhod topology on $[0,K]$.
\item Step 3: We show that $(\wh L_t^T)_t$ is tight for the Skorokhod topology on $[0,K]$.
\item Step 4: We identify uniquely the limit of distribution of a limit of $(\wh q_t^T)_t$ and we conclude with the convergence in distribution of $(\wh q_t^T)_t$.
\end{itemize}

\textbf{Step 1. } The upwards jumps from $q^T$ only come from the $L^T$ and therefore we have
\begin{equation*}
q_t^T - q^T_0 \leq  \int_0^t \int_0^\infty \1_{\{z \leq \lambda^{L,T}_{s}\}} \1_{q_{s-}^T \geq 0}  \pi^{L,T}(ds,  dz).
\end{equation*}
But $\lambda^{L}$ is non-increasing and therefore $\lambda^{L,T}_{s} \leq \beta_T \lambda^{L}(0)$ whenever $q_{s-}^T \geq 0$. Thus we have
\begin{equation}
\label{eq:upperboundq}
q_t^T - q^T_0 \leq \int_0^t \int_0^\infty \1_{\{z \leq \beta_T \lambda^{L}(0)\}}  \pi^{L,T}(ds,  dz).
\end{equation}
The process on the right-hand side is a Poisson process with deterministic constant intensity $\beta_T \lambda^{L}(0)$ and therefore for all $\varepsilon > 0$, there exists $L > 0$ such that 
\begin{equation*}
\mathbb{P} \bigg(\int_0^{TK} \int_0^\infty \1_{\{z \leq \beta_T \lambda^{L}(0)\}}  \pi^{L,T}(ds,  dz) \geq LT\beta_T \bigg) \leq \varepsilon
\end{equation*}
%
where $K > 0$. Recall that $\wh q^T_t = (T\beta_T)^{-1} q^T_{tT}$ and that $\wh q^T_0$ converges in distribution. Using the same notations as before, we deduce that 
\begin{equation*}
{\mathbb{P}} \Big(\sup_{t \leq K} \wh q_t^T \geq L\Big) \leq \varepsilon.
\end{equation*}

\textbf{Step 2. } We now show that $\overline{C}^T$ is tight using Aldous' criteria stated in Theorem VI.3.26 in \cite{jacod1987limit}. More precisely, we need to prove the following two conditions
\begin{enumerate}[label={(\roman*)}]
\item For all $\varepsilon > 0$, there exists $T_0 > 0$ and $A > 0$ such that for all $T > T_0$, 
\begin{equation*}
\mathbb{P}
\Big[
\sup_{t\leq K}
|
\wb C^T_t
|
> A
\Big] \leq \varepsilon.
\end{equation*}
\item We write $\mathcal{T}_K^T$ for the set of $(\mathcal{F}_{tT})_t$-stopping times $S$ bounded by $K$. For all $\varepsilon > 0$, we have
\begin{equation*}
\lim\limits_{\delta \to 0}
\limsup\limits_{T \to \infty}
\sup_{R,S \in \mathcal{T}_K, R \leq S \leq R + \delta}
\mathbb{P}
 [
 |
\wb C^T_S - \wb C^T_R
 |
\geq \varepsilon
]
=
0.
\end{equation*}
\end{enumerate}

We first prove (i). Take $\varepsilon > 0$. By \eqref{eq:upperboundq}, we know that for all $\eta > 0$, there exists $L > 0$ such that 
\begin{equation*}
{\mathbb{P}}\Big( \sup_{t \leq TL} q_t^T \geq LT\beta_T\Big) \leq \eta.
\end{equation*}
Therefore, for all $A > 0$, we have
\begin{equation*}
{\mathbb{P}} \Big[
\sup_{t\leq K}
|
\wh C^T_t
|
> A
\Big]
\leq
{\mathbb{P}} \Big[
C^T_{TK}
> T\beta_T A,\;
\sup_{t \leq TL} q_t^T \leq LT\beta_T
\Big]
+ \eta.
\end{equation*}
But since $\lambda^C$ is increasing, we know that when $\sup_{t \leq TL} q_t^T \leq LT\beta_T$, we have by Assumption \ref{assumption:scaling:lambda}
\begin{equation*}
\lambda_t^{C, T} = \lambda^{C,T}(q_{t-}^{T}) \leq \lambda^{C,T}(LT\beta_T) = \beta_T \lambda^{C}(L)
\end{equation*}
and therefore
\begin{equation*}
C^T_{TK} \leq \int_0^{TK} \int_0^\infty \1_{\{z \leq \beta_T \lambda^{C}(L)\}} \pi^{C,T}(ds dz).
\end{equation*}
The variable on the right-hand side follows a Poisson distribution with parameter $T \beta_T K \lambda^{C}(L)$ and therefore we obtain by Markov's inequality
\begin{equation*}
{\mathbb{P}} \Big[
\sup_{t\leq K}
|
\wh C^T_t
|
> A
\Big]
\leq
\frac{K \lambda^C(L)}{A}
+ \eta
\end{equation*}
which ensures that (i) holds by taking $\eta = \varepsilon /2$ and $A$ large enough.

We now focus on (ii). Proceeding similarly, we have
\begin{equation*}
{\mathbb{P}}\Big [
 |
\wh C^T_S - \wb C^T_R
 |
\geq \varepsilon
\Big]
\leq
\mathbb{P}
\bigg [
\int_{ST}^{(S+\delta)T} \int_0^\infty \1_{\{z \leq \beta_T \lambda^{C}(L)\}} \pi^{C,T}(ds dz)
\geq \varepsilon T\beta_T
\bigg] + \eta.
\end{equation*}
Since $S$ is a stopping time, $\int_{ST}^{(S+\delta)T} \int_0^\infty \1_{\{z \leq \beta_T \lambda^{C}(L)\}} \pi^{C,T}(ds dz)$ follows a Poisson distribution with parameter $T\beta_T \delta \lambda^C(L)$ and therefore using Markov's inequality
\begin{equation*}
{\mathbb{P}}\Big [
\big |
\wh C^T_S - \wb C^T_R
\big |
\geq \varepsilon
\Big]
\leq
\frac{\delta \lambda^C(L)}{\varepsilon} + \eta
\end{equation*}
and we conclude by taking $\delta \to 0$ and $\eta \to 0$.

%

\textbf{Step 3.} We already know from Step 2 that $\wh C^T$ is tight. Moreover, $\wh N^T$ is also tight by Proposition \ref{prop:convergence:jaisson} so the sum $\wh C^T + \wh N^T$ is also tight. We deduce that it is bounded in probability and therefore $\wb q^T$ is bounded below in probability. We can then repeat the same proof as in Step 2, utilizing this time the fact that $\lambda^L$ is decreasing and that $\wb q^T$ is bounded below. We conclude that $\wb L^T$ is also tight.

\textbf{Step 4.} We already know from steps 2 and 3 and from Proposition \ref{prop:convergence:jaisson} that $\wh C^T$, $\wh L^T$ and $\wh N^T$ are marginally tight. This implies the joint tightness of $(\wb C^T, \wb L^T, \wb N^T)$ and it only remains to identify uniquely the distribution of a limit. 

We consider $(C, L, X)$ a limit in distribution of a sub-sequence of $(\wh C^T, \wh L^T, \wh N^T)$. In the following, we do not write the subsequence indexes to ease the expressions. Without loss of generality, we can always suppose this convergence holds almost surely in the Skorokhod topology. Moreover, the jump size of $\wh C^T$, $\wh L^T$ and $\wh N^T$ is bounded by $(T\beta_T)^{-1} \to 0$ which ensure by Theorem VI.3.26 of Jacod and Shiryaev \cite{jacod1987limit} that $\wh C^T$, $\wh L^T$ and $\wh N^T$ are C-tight, meaning that the limit must be continuous almost surely. Therefore, the convergence $(\wh C^T, \wh L^T, \wh N^T) \to (C, L, X)$ holds almost surely for the sup-norm on $[0,K]$. In particular, we have
\begin{equation*}
\sup_{0 \leq t \leq K} | \wh q^T_t - q_t | \to 0 \text{ a.s.}\quad \text{with} \quad q_t = L_t - C_t - X_t.
\end{equation*}
Using that $\lambda^{L}$ and $\lambda^C$ are continuous, we know that
\begin{equation*}
\wh \lambda_t^{L, T} = \lambda^{L} (\wh q_t^T) \to \lambda^{L} (q_t)
\end{equation*}
almost surely for the sup-norm on $[0,K]$, which implies that the same convergence holds for
\begin{equation*}
\int_0^t \wh \lambda_s^{L, T}  ds \to \int_0^t \lambda^{L} (q_s)  ds.
\end{equation*}
We then use Doob's inequality on the martingales
$\wh L^{T}_t - \int_0^t \wh \lambda_s^{L, T}  ds$ and 
$\wh C^{T}_t - \int_0^t \wh \lambda_s^{C, T}  ds$ and the fact that their quadratic variation is respectively $(T\beta_T)^{-1} \wh L^{T}_t$ and $(T\beta_T)^{-1} \wh C^{T}_t$. This yields
\begin{equation*}
\mathbb{E}\bigg[{\sup_{t \leq T} \bigg|\wh L^{T}_t - \int_0^t \wh \lambda_s^{L, T}  ds \bigg|^2}\bigg]
\leq
(T\beta_T)^{-2} \mathbb{E}[( \wh L^{T}_K ) ^2]
\to 0
\end{equation*}
and
\begin{equation*}
\mathbb{E}\bigg[{\sup_{t \leq T} \bigg|\wh C^{T}_t - \int_0^t \wh \lambda_s^{C, T}  ds \bigg|^2}\bigg]
\leq
(T\beta_T)^{-2} \mathbb{E}[( \wh C^{T}_K ) ^2]
\to 0.
\end{equation*}
Consequently, using also Proposition \ref{prop:convergence:jaisson} to identify the distribution of $X$, we must have
\begin{equation*}
q_t = \int_0^t \lambda^{L}(q_s) - \lambda^C(q_s) - Y_s ds.
\end{equation*}
This imply that $q_t$ is continuously differentiable and its derivative is given by
\begin{equation*}
q_t' = \lambda^{L}(q_t) - \lambda^C(q_t) - Y_t.
\end{equation*}
Since $\lambda^{L} - \lambda^C$ is Lipschitz continuous, the solution must be unique so the distribution of $q_t$ is known. This also identifies uniquely $L$ and $C$, which conclude the proof.

\section{Proof of Theorem \ref{thm:scalingmi}}
\label{sec:proof:thm:scalingmi}

\subsection{Outline of the proof}

We consider three processes $X^a$, $\wb q^{a,t}$ and $q^{a}$ defined respectively by $X^{a}_s = \int_0^s Y^{a}_u du$ where $Y^{a}$ is given by \eqref{eq:def:Y},
 \begin{equation*}
q^{a,t}_s = q^{a}_0 + \int_0^s \lambda^{L}(q^{a}_u) du - \int_0^s \lambda^{C}(q^{a}_u) du - X^{a}_s
\end{equation*}
and 
\begin{equation*}
\wb q^{a,t}_s = q^{a}_0 + \int_0^s \lambda^{L}(\wb q^{a,t}_u) du - \int_0^s \lambda^{C}(\wb q^{a,t}_u) du - X^{a}_s + F^t(s)
\end{equation*}
where $F^t(s) = \int_0^{s\wedge t} f(u)  du$.

In this section, we plan to prove the convergence \begin{equation*}
\E [\wh{\mathrm{MI}}_t^T] = \E\bigg[
	\int_0^\infty \big( \kappa(\check q^{a,T,t}_s) - \kappa(\wh q^{a,T}_s) \big)  d \wh N^{a,T}_s 
\bigg] \to 
\E[\mathrm{MI}_t] = \E\bigg[
	\int_0^\infty \big( \kappa(\wb q^{a,t}_s) - \kappa( q^{a}_s) \big)  d X^{a}_s 
\bigg]
\end{equation*}
for $t>0$ fixed. For $A > 0$, we first decompose
\begin{equation*}
\E [\wh{\mathrm{MI}}_t^T] = \E\bigg[
	\int_0^A \big( \kappa(\check q^{a,T,t}_s) - \kappa(\wh q^{a,T}_s) \big)  d \wh N^{a,T}_s 
+
	\int_A^\infty \big( \kappa(\check q^{a,T,t}_s) - \kappa(\wh q^{a,T}_s) \big)  d \wh N^{a,T}_s 
\bigg]
\end{equation*}
and
\begin{equation*}
\E[\mathrm{MI}_t] = \E\bigg[
	\int_0^A \big( \kappa(\wb q^{a,t}_s) - \kappa( q^{a}_s) \big)  d X^{a}_s 
+
	\int_A^\infty \big( \kappa(\wb q^{a,t}_s) - \kappa( q^{a}_s) \big)  d X^{a}_s 
\bigg].
\end{equation*}

Each of these terms are studied controlled in the following lemmas, proved in Sections \ref{sec:proof:lem:completion:1}, \ref{sec:proof:lem:completion:2} and \ref{sec:proof:lem:completion:3} respectively.

\begin{lemma}
\label{lem:completion:1}
For all $\varepsilon > 0$, there exists $A > t$ such that for all $T$, we have
\begin{equation*}
\bigg|
\E\bigg[
	\int_A^\infty \big( \kappa(\check q^{a,T,t}_s) - \kappa(\wh q^{a,T}_s) \big)  d \wh N^{a,T}_s 
\bigg]
\bigg|
\leq \varepsilon.
\end{equation*}
\end{lemma}

\begin{lemma}
\label{lem:completion:2}

For all $\varepsilon > 0$, there exists $A > t$ such that
\begin{equation*}
\bigg|
\E\bigg[
	\int_A^\infty \big( \kappa(\wb q^{a,t}_s) - \kappa( q^{a}_s) \big)  d X^{a}_s 
\bigg]
\bigg|
\leq \varepsilon.
\end{equation*}
\end{lemma}

\begin{lemma}
\label{lem:completion:3}

For all $A > t$ we have
\begin{equation*}
\E\bigg[
	\int_0^A \big( \kappa(\check q^{a,T,t}_s) - \kappa(\wh q^{a,T}_s) \big)  d \wh N^{a,T}_s 
\bigg]
\to
 \E\bigg[
	\int_0^A \big( \kappa(\wb q^{a,t}_s) - \kappa( q^{a}_s) \big)  d X^{a}_s 
	\bigg]
\end{equation*}
when $T \to \infty$.
\end{lemma}

The proof of Theorem \ref{thm:scalingmi} is a consequence of these lemmas. In fact, fix $\varepsilon > 0$. Then combining lemmas \ref{lem:completion:1} and \ref{lem:completion:2}, for $\varepsilon > 0$, there exists $A > t$ such that for all $T$, we have
\begin{equation*}
    \bigg|
\E\bigg[
	\int_A^\infty \big( \kappa(\check q^{a,T,t}_s) - \kappa(\wh q^{a,T}_s) \big)  d \wh N^{a,T}_s 
\bigg]
\bigg|
+
    \bigg|
\E\bigg[
	\int_A^\infty \big( \kappa(\wb q^{a,t}_s) - \kappa( q^{a}_s) \big)  d X^{a}_s 
\bigg]
\bigg|
\leq \frac{1}{2} \varepsilon.
\end{equation*}
By Lemma \ref{lem:completion:3}, we can choose $T_0$ large enough so that for all $T \geq T_0$, we have
\begin{equation*}
\bigg|
    \E\bigg[
	\int_0^A \big( \kappa(\check q^{a,T,t}_s) - \kappa(\wh q^{a,T}_s) \big)  d \wh N^{a,T}_s 
\bigg]
-
 \E\bigg[
	\int_0^A \big( \kappa(\wb q^{a,t}_s) - \kappa( q^{a}_s) \big)  d X^{a}_s 
	\bigg]
\bigg|
\leq \frac{1}{2} \epsilon.
\end{equation*}
Combining these inequalities, we see that
\begin{equation*}
    \forall \epsilon > 0,  \exists T_0 > 0,  \forall T \geq T_0,  
    |
    \E [\wh{\mathrm{MI}}_t^T]
-
 \E [\mathrm{MI}_t]
|
\leq \epsilon,
\end{equation*}
which proves the desired convergence.

\subsection{Proof of Lemma \ref{lem:completion:1}}
\label{sec:proof:lem:completion:1}

By definition, and using that the compensator of $N^{a,T}$ is $\int_0^\cdot \lambda^{a,T}_s ds$, we have
\begin{align*}
\E\bigg[
	\int_A^\infty \big( \kappa(\check q^{a,T,t}_s) - \kappa(\wh q^{a,T}_s) \big)  d \wh N^{a,T}_s 
	 \bigg]
&=
\frac{1}{T\beta_T}
\E\bigg[
	\int_{AT}^\infty \big( \kappa^T (\wb q^{a,T,t}_s) - \kappa^T( q^{a,T}_s) \big)  d N^{a,T}_s 
	 \bigg]
\\
&=
\frac{1}{T\beta_T}
\E\bigg[
	\int_{AT}^\infty \big( \kappa^T (\wb q^{a,T,t}_s) - \kappa^T( q^{a,T}_s) \big) \lambda^{a,T}_s  ds 
	 \bigg].
\end{align*}
Moreover, $\kappa^T$ is Lipschitz continuous and $|\kappa^T|_{lip} = (T\beta_T)^{-1}|\kappa|_{lip}$ by Assumption \ref{assumption:kappa}. Thus, we get
\begin{equation*}
\bigg|
\E\bigg[
	\int_A^\infty \big( \kappa(\check q^{a,T,t}_s) - \kappa(\wh q^{a,T}_s) \big)  d \wh N^{a,T}_s 
	 \bigg]
	 \bigg|
\leq
\frac{|\kappa|_{lip}
}{(T\beta_T)^2}
\E\bigg[
	\int_{AT}^\infty |\wb q^{a,T,t}_s - q^{a,T}_s| \lambda^{a,T}_s  ds 
	 \bigg].
\end{equation*}

Denote by $\wt \lambda^T$ the function defined by 
\begin{equation*}
\wt \lambda^T(k) = \inf_{q \in \bbZ} \big\{ \lambda^{L,T}(q) - \lambda^{L,T}(q + k) + \lambda^{C,T}(q + k) - \lambda^{C,T}(q) \big\}.
\end{equation*}
Following the proof of Lemma \ref{lem:expvar}, we see that there exists a Poisson point measure $\wt \pi$ independent of $\lambda^{a,T}$ such that the process $Y^T$ defined for $s \geq t$ by 
\begin{equation*}
Y^T_s = N^o_t - \int_0^t \int_0^\infty \1_{z \leq \wt \lambda^T(Y^T_{u-})} \wt \pi(du, dz)
\end{equation*}
satisfies $|\wb q^{a,T,t}_s - q^{a,T}_s| \leq Y^T_s$ for all $s \geq t$. Therefore, we have
\begin{align*}
\bigg|
\E\bigg[
	\int_A^\infty \big( \kappa(\check q^{a,T,t}_s) - \kappa(\wh q^{a,T}_s) \big)  d \wh N^{a,T}_s 
	 \bigg]
	 \bigg|
&\leq
\frac{|\kappa|_{lip}
}{(T\beta_T)^2}
\E\bigg[
	\int_{AT}^\infty Y^T_s \lambda^{a,T}_s  ds 
\bigg]
\\
&\leq
\frac{|\kappa|_{lip}
}{(T\beta_T)^2}
\E\bigg[
	\int_{AT}^\infty
	\E[Y^T_s]
	\E[\lambda^{a,T}_s] ds
\bigg].
\end{align*}

Using Assumptions \ref{assumption:scaling:lambda} and \ref{assumption:mean_rev}, we know that for $k \geq 0$, we have $\wt \lambda^T(k) \geq c k T^{-1}$ and therefore we have
\begin{equation*}
\E[Y^T_s |  \calF_{tT}^T] \leq 
N^o_t - c T^{-1} \int_{tT}^s \E[Y^T_u |  \calF_{tT}^T] du
\end{equation*}
which implies by Grönwall that 
\begin{equation*}
\E[Y^T_s |  \calF_{tT}^T] \leq N^o_t \exp\big( -cT^{-1} (s - tT)\big).
\end{equation*}
Moreover, it is shown in \cite{jaisson2016rough} that $(\beta_T)^{-1} 	\E[\lambda^{a,T}_s]$ is bounded uniformly in $s$ and $T$. Combining these bounds and using also that $\E[N^{o,T}_{tT}] = T\beta_T \int_0^t f(s) ds$, we see that there exists $C > 0$ such that
\begin{align*}
\bigg|
\E\bigg[
	\int_A^\infty \big( \kappa(\check q^{a,T,t}_s) - \kappa(\wh q^{a,T}_s) \big)  d \wh N^{a,T}_s 
	 \bigg]
	 \bigg|
&\leq
\frac{C}{T^2\beta_T}
\E\bigg[
\int_{AT}^\infty
N^o_t \exp\big( -cT^{-1} (s - tT)\big)
 ds
\bigg]
\\
&\leq
\frac{C}{T\beta_T}
\E\bigg[
\int_{A-t}^\infty
N^{o,T}_{tT} \exp( -cs)
 ds
\bigg]
\\
&\leq
\frac{C}{c} \exp\big(-c(A-t)\big) \int_0^t f(s) ds.
\end{align*}
Clearly, $\exp(-c(A-t)) \to 0$ as $A \to \infty$, which concludes the proof of Lemma \ref{lem:completion:1}.

\subsection{Proof of Lemma \ref{lem:completion:2}}
\label{sec:proof:lem:completion:2}

Following the same approach as for Lemma \ref{lem:completion:2}, we have
\begin{align*}
\bigg|
\E\bigg[
	\int_A^\infty \big( \kappa(\wb q^{a,t}_s) - \kappa(q^{a}_s) \big)  d X_s^{a} 
	 \bigg]
	 \bigg|
&\leq
|\kappa|_{lip}
\E\bigg[
	\int_{A}^\infty |\wb q^{a,t}_s - q^{a}_s| Y^{a}_s  ds 
\bigg].
\end{align*}
First recall that $q^{a}$ and $\wt q^{a,t}$. and we have
\begin{equation*}
\begin{cases}
q^{a\prime}_s
=
\lambda^L(\wb q^{a,t}_s) - \lambda^C(\wb q^{a,t}_s)
- Y^{a}_s,
\\
\wb q^{a,t\prime}_s
=
\lambda^L(q^{a}_s) - \lambda^C(q^{a}_s) 
- Y^{a}_s
+ f(s) \1_{s \leq t}.
\end{cases}
\end{equation*}
Comparison of ordinary differential equations ensure that $\wb q^{a,t\prime}_s \geq q^{a\prime}_s$ for all $s \geq 0$. Moreover, using Assumption \ref{assumption:mean_rev}, we have for all $s \geq 0$
\begin{equation*}
\wb q^{a,t\prime}_s - q^{a\prime}_s \leq -c(\wb q^{a,t}_s - q^{a}_s) + f(s) \1_{s \leq t}.
\end{equation*}
This implies in particular that
\begin{equation*}
|\wb q^{a,t}_t - q^{a}_t| \leq \int_0^t f(s)  ds
\end{equation*}
and by Grönwall's inequality, we have
\begin{equation*}
\wb q^{a,t\prime}_s - q^{a\prime}_s \leq \int_0^t f(s)  ds \exp \big( -c (s-t)\big)
\end{equation*}
for all $s \geq t$.

On the other hand, we have
\begin{equation*}
\E[Y^a_s] = F^{\alpha, \lambda}(s) \leq \int_0^\infty E_{\alpha, \alpha}(-u) du
\end{equation*}
which is finite by definition of the Mittag-Leffler function.

We conclude by combining these inequalities, with the same arguments as for Lemma \ref{lem:completion:1}.

\subsection{Proof of Lemma \ref{lem:completion:3}}
\label{sec:proof:lem:completion:3}

Using Propositions \ref{prop:convergence:jaisson},  \ref{prop:convergence:LOB} and \ref{prop:convergence:LOBMO}, we know that on for all $A \geq 0$,
\begin{equation*}
\check q^{a,T,t} \to \wb q^{a,t},
\quad\quad
\wh q^{a,T} \to q^{a}
\quad\text{ and }\quad
 N^{a,T} \to X^{a}.
\end{equation*}
in distribution, for the Skorokhod topology on $[0,A]$. Without loss of generality, we can assume that this convergence holds almost surely. Since all the limiting processes are continuous, the convergence holds for the sup-norm on $[0,A]$. Using the continuity of the limiting process, we deduce that 
\begin{equation*}
	\int_0^A \big( \kappa(\check q^{a,T,t}_s) - \kappa(\wh q^{a,T}_s) \big)  d \wh N^{a,T}_s 
\to
	\int_0^A \big( \kappa(\wb q^{a,t}_s) - \kappa( q^{a}_s) \big)  d X^{a}_s
\end{equation*}
almost surely. It remains to prove that this also holds in $L^1$. To do so, we plan to prove that $\int_0^A ( \kappa(\check q^{a,T,t}_s) - \kappa(\wh q^{a,T}_s) )  d \wh N^{a,T}_s $ is uniformly integrable by proving that it is bounded in $L^2$. By definition, we have
\begin{align*}
	\int_0^A \big( \kappa(\check q^{a,T,t}_s) - \kappa(\wh q^{a,T}_s) \big)  d \wh N^{a,T}_s 
	=
	&\frac{1}{T\beta_T}
	\int_0^{TA} \big( \kappa^T(\wb q^{a,T,t}_s) - \kappa^T(q^{a,T}_s) \big)  d M^{a,T}_s 
\\&+
	\frac{1}{T\beta_T}
	\int_0^{TA} \big( \kappa^T(\wb q^{a,T,t}_s) - \kappa^T(q^{a,T}_s) \big) \lambda^{a,T}_s  ds
\end{align*}
where $M^{a,T}_s = N^{a,T}_s - \int_0^s \lambda^{a,T}_u  du$. Note that $M^{a,T}$ is a martingale whose quadratic variation is $N^{a, T}$. Therefore, using also that $\kappa^T$ is positive and bounded uniformly in $T$ by $\norm{\kappa}_\infty$, we have
\begin{equation*}
	\E\bigg[\bigg(\int_0^A \big( \kappa(\check q^{a,T,t}_s) - \kappa(\wh q^{a,T}_s) \big)  d \wh N^{a,T}_s \bigg)^2 \bigg]
	\leq
	\frac{2 \norm{\kappa}_\infty }{(T\beta_T)^2}
\bigg(	
	\E[N^{a,T}_{TA}] 
	+ 
	\E\bigg[\bigg(\int_0^{TA} \lambda^{a,T}_s  ds \bigg)^2 \bigg]
\bigg).
\end{equation*}
We already know that $\E[N^{a,T}_{TA}]$ is of order $T\beta_T$ and we have
\begin{equation*}
\E\bigg[\bigg(\int_0^{TA} \lambda^{a,T}_s  ds \bigg)^2 \bigg]
\leq T^2 A^2 \sup_{0 \leq s \leq AT} \E[(\lambda^{a,T}_s)^2].
\end{equation*}
Using Lemma \ref{lemma:m2}, we see that $ \sup_{0 \leq s \leq AT} \E[(\lambda^{a,T}_s)^2]$ is of order $(\beta_T)^2$, which concludes the proof.

\end{document}